\documentclass[a4paper]{article}
\usepackage[T1]{fontenc}
\usepackage[utf8]{inputenc}
\usepackage{graphicx}
\usepackage{hyperref}
\usepackage{amssymb}
\usepackage{amsmath,bm}
\usepackage{lmodern} 
\usepackage{textcomp} 
\usepackage{sectsty}
\usepackage{fancybox}
\usepackage{listings}
\usepackage{charter}
\usepackage{epsfig}
\usepackage{tabularx}
\usepackage{latexsym}
\usepackage{alltt}
\usepackage{setspace}
\usepackage{datetime}
\newdateformat{mydate}{\monthname[\THEMONTH], \THEYEAR}
\usepackage[ruled,vlined]{algorithm2e}
\usepackage{caption}
\usepackage{enumitem}
\def\qed{\hfill$\Box$}

\usepackage[showframe=false, left=3cm, top=2cm, right=3cm, bottom=2cm ]{geometry}
\usepackage{appendix}
\usepackage{authblk}

\def\qed{\hfill$\Box$}

\newtheorem{corollary}{Corollary~}

\newtheorem{theorem}{Theorem~}
\newtheorem{proposition}{Proposition~}

\newtheorem{observation}{Observation~}

\newtheorem{lemma}{Lemma~}
\def\abstract{{\begin{center}
\Large {\bf Abstract}
\end{center} }}

\makeatletter

\def\@seccntformat#1{\@ifundefined{#1@cntformat}%
   {\csname the#1\endcsname\quad}  
   {\csname #1@cntformat\endcsname}
}
\makeatother
\begin{document}
\title{Bicoloring covers for graphs and hypergraphs}
\author{Tapas Kumar Mishra
\qquad
Sudebkumar Prasant Pal
         }
\affil{ Department of Computer Science and Engineering\\ IIT Kharagpur 721302 \\ India}
         
\date{}
\maketitle

\begin{abstract}
Let the {\it bicoloring cover number $\chi^c(G)$} for a hypergraph $G(V,E)$ be the 
minimum number of bicolorings of vertices of $G$ such 
that every hyperedge $e\in E$ of $G$
is properly bicolored in at least one of the $\chi^c(G)$ bicolorings. 
We establish a tight bound for $\chi^c(K_n^k)$, 
where  $K_n^k$ is the complete $k$-uniform hypergraph on $n$ vertices.
We investigate the relationship between $\chi^c (G)$,  
matchings,  hitting sets, $\alpha(G)$(independence number) 
and $\chi(G)$  (chromatic number). 
We design a factor $O(\frac{\log n}{\log \log n-\log \log \log n})$ 
approximation algorithm for computing a bicoloring 
cover.
We define a new parameter for 
hypergraphs - "cover independence number $\gamma(G)$" and 
prove that $\log \frac{|V|}{\gamma(G)}$ and $\frac{|V|}{2\gamma(G)}$ are 
lower bounds for $\chi^c(G)$ and $\chi(G)$, respectively. 
We show that $\chi^c(G)$ can be approximated by a 
polynomial time algorithm achieving approximation ratio  
$\frac{1}{1-t}$, 
if $\gamma(G)=n^t$, where $t<1$. 
We also construct a particular class of hypergraphs $G(V,E)$ 
called {\it cover friendly} hypergraphs 
where the ratio of $\alpha(G)$ to $\gamma(G)$ can be
arbitrarily large.
We prove that for any $t\geq 1$, 
there exists a $k$-uniform hypergraph $G$ 
such that the {\it clique number} $\omega(G)=k$ and $\chi^c(G) > t$.
Let $m(k,x)$ denote the minimum number of hyperedges 
such that some $k$-uniform hypergraph $G$ with $m(k,x)$ 
hyperedges does not have a 
bicoloring cover of size $x$. 
We show that $ 2^{(k-1)x-1} < m(k,x) \leq x \cdot k^2 \cdot 2^{(k+1)x+2}$.
Let the {\it dependency $d(G)$} of $G$ be the 
maximum number of hyperedge neighbors of any 
hyperedge in $G$. 
We propose an algorithm for computing a 
bicoloring cover of size $x$ for $G$ if 
$d(G) \leq(\frac{2^{x(k-1)}}{e}-1)$ 
using 
$nx+kx\frac{m}{d}$ 
random bits.
\end{abstract}

\bigskip

\section{Introduction}
\label{sec:intro}

We define the {\it bicoloring cover number} $\chi^c(G)$ for a hypergraph $G(V,E)$
as the 
minimum number of bicolorings such that every hyperedge $e\in E$ of $G$
is properly bicolored in at least one of the $\chi^c$ bicolorings.
Let ${\cal X}$ be a set of bicolorings $\{X_1,X_2,...,X_t\}$. 
Then ${\cal X}$ is a bicoloring cover for $G$ if for each hyperedge $e$
of $G$, there is an integer $i\in \{1,2,...t\}$, such that $e$ is non-monochromatic 
with respect to bicoloring $X_{i}$.

Consider the scenario where $n$ doctors can each 
be assigned one 
of two kinds of {\it tasks}; either he can see patients or 
perform operations. All doctors are equivalent and can 
perform only one of the two tasks in each {\it group}. 
There are $m$ groups made from this set of $n$ doctors viz., 
$E_1,E_2,...,E_m$, where each group is of size $k$. 
Any doctor can be a member of multiple groups. 
In order to provide proper treatment, all the $k$ members 
of no group should be assigned the same task; each group must 
have at least one doctor seeing patients and at least one doctor
performing operations. 
Given $n$ doctors and $m$ {\it groups} of doctors, 
is there a possible allocation of {\it tasks} to doctors such that 
none of the groups has all doctors allocated the same
task? This problem can viewed as the hypergraph bicoloring 
problem for the $k$-uniform hypergraph $G(V,E)$, where $n=|V|$, $m=|E|$.
Here, the doctors represent vertices, the groups represent 
$k$-uniform hyperedges, and the tasks assigned to doctors represent the 
two colors for bicoloring vertices. However, there exist hypergraphs
that are not bicolorable. For such hypergraphs, it makes sense to use a
bicoloring cover with $\chi^c$ bicolorings. Instead of all $m$ groups 
of doctors being deployed simultaneously, we could have a minimum number $\chi^c$ of
deployments, one for each of the bicolorings
from a bicoloring cover for $G(V,E)$. 
Note that in any of these bicolorings,
the same doctor can serve in multiple groups. 
Observe that if we have to deploy each of the $m$ groups of
doctors effectively, then we need at least $\chi^c$ bicolorings, where each
bicoloring yields one shift of duty assignments.
The minimum number of shifts required for deploying
all the $m$ groups of doctors,
is therefore the 
bicoloring cover number $\chi^c(G)$. 
Throughout the paper, $G$ denotes a $k$-uniform hypergraph with vertex set $V$ and hyperedge set $E$, unless otherwise stated. We use $V(G)$ and $V$, and $E(G)$ and $E$ interchangeably.
All logarithms are to the base two unless specified 
otherwise.

\subsection{Related works}
Graph decomposition is a widely studied problem in graph theory. 
The main idea of the problem is whether a given graph $G(V,E)$
can it be decomposed into some family of smaller graphs i.e., 
is there a family of graphs ${\cal{H}}=\{H_1,...,H_j\}$ such that
(1). $V(H_i) \subseteq V(G)$ for all ${H_i \in {\cal{H}}}$,
(2). $\cap_{H_i \in \cal{H}} E(H_i)=\phi$ and
(3). $\cup_{H_i \in \cal{H}} E(H_i)=E(G)$.
In other words, the family of graphs ${\cal{H}}$ covers $G$,
or partitions the edge set of $G$. 
If such a ${\cal{H}}$ exists, then splitting $G$ 
into $\{H_1,...,H_j\}$ is called a ${\cal{H}}$-decomposition of $G$. 
A kind of decomposition studied requires ${\cal{H}}$ to a 
single graph (say $\{H1\}$) and checks if $G$ 
can be decomposed into multiple copies of $H1$ with 
the disjoint intersection condition omitted. 
Such a decomposition is denoted by $H1|G$. 
The family ${\cal{H}}$ may consist of paths, cycles, bipartite graphs or matchings.
For instance, consider {\it matching decomposition}, where 
in an edge-coloring of $G$, each color class is a matching. 
So, coloring edges of $G$ by $\chi_e(G)$ colors properly gives 
the minimum matching decomposition of the graph. 
Vizing's theorem \cite{West2000} states that for 
all simple graphs $G$, $\chi_e(G) \leq \Delta(G)+1$. 
As a result, there is always a matching decomposition 
of $G$ into ${\cal{H}}$ of 
size $|{\cal{H}}|=\Delta(G)+1$.
A $tK_2|G$ decomposition is splitting $G$ into multiple 
copies of  $t$ $K_2$'s i.e., matchings of size $t$. 
Bialostocki  and  Roditty \cite{BiaRodi1982} proved that $3K_2|G$ 
if and only if $3||E(G)|$ and $\Delta(G)\leq \frac{|E(G)|}{3}$, 
with a finite number of exceptions. 
Alon \cite{Alon1983} shown that for every $t>1$, 
if $|E(G)| \geq \frac{8}{3}t^2-2t$, $tK_2|G$ if and only 
if $t||E(G)|$ and $\Delta(G)\leq \frac{|E(G)|}{t}$.
Along similar lines, a significant amount of study has been done 
and there is vast literature for various kinds of decomposition of graphs 
(see \cite{ChuGra1981}). In this paper, we aim to combine the concepts of
decomposition and coloring graphs and hypergraphs.
%
%
%
%

\subsection{Our contribution}

We define $\chi^c(G)$ for a hypergraph $G(V,E)$
as the 
minimum number of bicolorings that guarantees every hyperedge $e\in E$ of $G$
is properly bicolored in at least one of the $\chi^c(G)$ bicolorings. 
In section \ref{sec:preliminaries}, (i) we derive a tight bound for $\chi^c(G)$ for the 
complete $k$-uniform hypergraph $G$,
(ii) establish upper bounds for $\chi^c (G)$
based on {\it matchings} and {\it hitting sets} of the hypergraph, 
and, (iii) design polynomial time algorithms for computing bicoloring covers. 
We also relate  $\chi^c(G)$ with independent sets and chromatic numbers and 
show that $\chi^c(G) = \lceil\log \chi(G) \rceil$.

In section \ref{sec:approx}, we present an inapproximability result about
the impossibility of approximating the bicoloring cover of  $n$-vertex 
$k$-uniform hypergraphs, to within an additive factor 
of $({1-\epsilon}) \log n $, for any fixed $\epsilon >0$ in time polynomial 
in $n$. 
For a $k$-uniform hypergraph $H(V,E)$, where $|V|=n$, 
we show that the bicoloring cover number 
$\chi^c(H)$ is $O(\frac{\log n}{\log \log n-\log \log \log n}) \allowbreak$ approximable. 

Let $C=\{C_i| C_i\text{ is a bicoloring cover of } \chi^c(G) 
\text{ bicolorings that cover }\allowbreak G(V,E)\}$.
Let $|C|=w$, where $w \leq 2^{n\chi^c(G)}$.
Each vertex receives a color bit vector due to a bicoloring cover $C_i$.
Let $\gamma_i(G)$ be the size of the largest set of vertices that receive 
the same color bit vector due to bicoloring cover $C_i$. 
$\text{Let }\gamma(G)=  \max_{\substack{
            1\leq i \leq w}}
     		\gamma_i(G)$.
We call $\gamma(G)$ the {\it cover independence} number of hypergraph $G$. 
In section \ref{sec:independentset}, we show that for any 
$k$-uniform hypergraph $G$, $\gamma(G) \geq k-1$. 
We relate $\gamma(G)$ to $\chi^c(G)$ and $\chi(G)$ and derive the 
lower bounds of $\log  \frac{|V|}{\gamma(G)}$ and 
$\frac{|V|}{2\gamma(G)}$ 
for $\chi^c(G)$ and $\chi(G)$, respectively.
We also construct a particular class of hypergraphs $G(V,E)$ 
called {\it cover friendly} 
hypergraphs where the ratio of $\alpha(G)$ to $\gamma(G)$ 
can be made arbitrarily large.
More specifically, we construct $k$-uniform hypergraphs 
$G(V,E)$ where $\alpha(G) \geq n^{1-t}$, whereas $\gamma(G)=n^t$, 
for some small fraction $0<t <0.5$. 
We show in 
Corollary \ref{corr:1} 
that 
$\chi^c(G)$ can be approximated for  
such cover friendly hypergraphs, 
with an 
approximation ratio of $\frac{1}{1-t}$, by exploiting the special 
properties of such hypergraphs. 
However, using Proposition \ref{prop:1}, 
we can only achieve an approximation ratio of at least $\frac{1}{t}$
for cover friendly hypergraphs. 
This implies that we achieve an improvement (reduction) in 
approximation ratio for estimating $\chi^c(G)$ for cover friendly hypergraphs
by a 
factor of at least $\frac{1-t}{t}$ by 
using the properties of $\gamma(G)$ in Corollary \ref{corr:1}; 
the approximation ratio achieved using Observation \ref{obs:independentsetrel}
in Proposition \ref{prop:1} is much smaller for cover friendly graphs.
Furthermore, our constant factor approximation ratio 
of $\frac{1}{1-t}$ for 
approximating $\chi^c(G)$ for cover friendly hypergraphs 
is in sharp contrast to our
$\allowbreak O(\frac{\log n}{\log \log n-\log \log \log n})$ factor algorithm
for estimating $\chi^c(G)$ for 
general hypergraphs as summarized in Theorem \ref{thm:sudakovapp} 
of Section \ref{sec:approx}.

Let $H(V',E')$ be the largest $k$-uniform subhypergraph of a 
$k$-uniform hypergraph $G(V,E)$, where $V' \subseteq V$, $E' \subseteq E$, 
and there is a hyperedge for every subset of $k$ vertices in $V'$ 
i.e., $E'=\binom{V'}{k}$. We define $\omega(G)=|V'|$.
In Section \ref{subsec:clique}, we prove that for any $t\geq 1$, 
there exists a $k$-uniform hypergraph $G$ 
where $\omega(G)=k$ and $\chi^c(G) > t$. Observe that, 
for $k=2$ (usual graphs), this result implies that triangle-free 
graphs can have arbitrarily large bicoloring cover numbers.

In sections \ref{sec:probabilistic} and \ref{sec:deprel}, we 
correlate $\chi^c(G)$ to the number $|E|$ of hyperedges, 
and the {\it dependency d(G)}, using probabilistic analysis, the Moser-Tardos 
algorithm~\cite{Moser:2010:CPG:1667053.1667060}, and 
an incremental method based on cuts in hypergraphs. 
We show that if $|E| \leq 2^{(k-1)x-1}$, then a bicoloring 
cover of size 
$x$ can be computed in polynomial time. 
We use $m(k,x)$ to denote the minimum number of hyperedges 
such that some $k$-uniform hypergraph $G$ with $m(k,x)$
hyperedges does not have a 
bicoloring cover of size $x$. 
We show that $2^{(k-1)x-1} < m(k,x) \leq x \cdot k^2 \cdot 2^{(k+1)x+2}$.
Let the {\it dependency $d(G)$} of $G$ be the maximum number of 
neighboring hyperedges of a hyperedge in $G$. 
We use the Moser-Tardos constructive approach 
for Lov\'{a}sz local lemma,
as in 
\cite{Moser:2010:CPG:1667053.1667060},  
for computing bicoloring covers of size $x$, where 
the dependency $d(G)$ of the hypergraph is 
bounded by
$\frac{2^{x(k-1)}}{e}-1$. 


\section{Preliminaries}
\label{sec:preliminaries}

\subsection{Bicoloring cover number and chromatic number for hypergraphs}
\label{subsec:chromatic}

We establish the following result relating $\chi^c(G)$ and $\chi(G)$ for arbitrary 
hypergraphs.

\begin{theorem}\label{thm:chromaticrelation}
Let $G(V,E)$ be a hypergraph. 
Let  $\chi^c(G)$ and $\chi(G)$ be the bicoloring cover number and 
chromatic number of $G$, respectively. Then,
$\chi^c(G)= \lceil \log \chi(G)\rceil$.
\end{theorem}
%


To show that $ \lceil \log \chi(G)\rceil\leq \chi^c(G)$, 
choose a bicoloring cover $C$ of size $\chi^c(G)$ for $G$. 
Each vertex $v$ of $G$ is assigned 
a set of 
$\chi^c(G)$ 
colors (bits 0 or 1), 
by the 
$\chi^c(G)$ 
bicolorings in the bicoloring cover $C$. 
Assign the decimal equivalent of the $\chi^c(G)$-bit pattern for $v$
as the color for $v$ to get a vertex-coloring $C'$ for $G$.  
The total number of colors used is at most $2^{\chi^c(G)}$. 
We claim $C'$ is a proper vertex-coloring for $G$, thereby enforcing the
inequality
$\chi(G)\leq 2^{\chi^c(G)}$ or 
$\lceil \log \chi(G) \rceil \leq \chi^c(G)$.  
For the sake of contradiction, assume that some hyperedge $e \in E(G)$ 
is monochromatic under $C'$. 
This means in each of the $\chi^c(G)$ bicolorings, 
every vertex of $e$ gets same color. As a result, $e$  
is not covered by the  $\chi^c(G)$ sized cover, which is a contradiction. 
Consequently, we have the following lemma.

\begin{lemma}
\label{lemma:ch1} 
For a hypergraph $G(V,E)$, 
$\lceil \log \chi(G) \rceil \leq \chi^c(G)$.
\end{lemma}

To prove the second inequality, consider a proper 
coloring $C$ of the vertices of $G$ with $\chi(G)$ colors. 
Construct the bicoloring cover $X$ 
of size $\lceil \log \chi(G) \rceil$ by assigning 
the vertices with two colors determined by the 0/1 bits of the 
color they were assigned under proper coloring $C$; 
a vertex $v$ is assigned the $i$th bit of the color assigned to it 
under coloring $C$ for the 0/1 bicoloring of $v$ in the $i$th 
bicoloring of the bicoloring cover $X$, 
$1\leq i\leq \lceil \log \chi(G) \rceil$. 
Assume for the sake of contradiction that some 
$e \in E(G)$ is not covered under bicoloring cover $X$. 
This means every vertex of $e$ has the same bit vector of length
$\lceil \log \chi(G) \rceil$, and therefore has the same color under 
coloring $C$, a contradiction. Consequently, we have the following lemma.

\begin{lemma}
\label{lemma:ch2}
For a hypergraph $G(V,E)$, 
$\chi^c(G)\leq \lceil \log \chi(G)\rceil$.
\end{lemma}

Theorem \ref{thm:chromaticrelation} follows
from Lemmas \ref{lemma:ch1} and \ref{lemma:ch2}. 
The following lemma is a direct consequence  of Theorem \ref{thm:chromaticrelation}.

\begin{lemma}\label{lemma:indpset}
Let $G(V,E)$ be a $k$-uniform hypergraph, 
and $\{I_1$, $I_2$,...,$I_{u}\}$ be a
partition of the vertex set $V$ into independent sets. 
Then there exists a bicoloring cover for $G$ of size $\lceil \log u\rceil$.
\end{lemma}


\subsection{Matchings, hitting sets and bicoloring covers for hypergraphs}
\label{subsec:matchhitcover}

Let $G(V,E)$ be a $k$-uniform hypergraph, 
with $|V|=n$ and $E=\{E_1,E_2,...,E_m\}$,
where $E_i\subseteq V$, $1\leq i\leq n$.
We have the following bounds for $\chi^c(G)$ based on the sizes of  
maximal matchings and hitting sets. The first algorithm \ref{algo:matching} uses a 
maximal matching for computing a bicoloring cover. The second algorithm 
\ref{algo:hittingset} uses a hitting set.

\begin{algorithm}[H]
\SetAlgoRefName{$MBC$}
\KwData{$k$-uniform hypergraph $G(V,E)$ with $|V|=n$, and 
a maximal matching $M$ of $G$}
 \KwResult{Set $X$ of bicolorings of size $|X| \leq log_2 |M| +2 $}
 
Color every vertex in the hyperedges of $M$ with color 1 and rest of the vertices with color 2\;
$recMBC$($M$)\;
Color the remaining hyperedge of the matching independently using one bicoloring\;
\caption{Computing bicoloring cover using a Maximal matching $M$}
\label{algo:matching}
\end{algorithm}

\begin{function}
\KwIn{A set of hyperedges $M$}
\If{$(|M|>1)$}{
Split the hyperedges in $M$ into two sets $A,B$ of size $\lfloor\frac{|M|}{2}\rfloor$ and $\lceil\frac{|M|}{2}\rceil$ respectively\;
Color every vertex in $A$ with color 1 and every vertex in $B$ with color 2\;

$recMBC(A)$;

$recMBC(B)$;

}

\caption{$recMBC$($M$)}
\end{function}

%
%

Let $M$ be a maximal matching of the $n$-vertex $k$-uniform 
hypergraph $G(V,E)$.
We propose an algorithm \ref{algo:matching} for computing a 
bicoloring cover of $G$ using $M$.
The algorithm \ref{algo:matching} takes the hypergraph $G(V,E)$ and a maximal 
matching $M$ of $G$
as inputs and produces a bicoloring cover $C1$ for $G$.
Let $V_M$ denote the set of vertices in the hyperedges in $M$.
In the first bicoloring, \ref{algo:matching} colors every 
vertex of $V_M$ with color 0, and, all the vertices in $V \setminus V_M$ 
with color 1.
Due to the maximality of the matching $M$, 
every hyperedge that contains a vertex from 
$V \setminus V_M$
shares at
least one vertex with some hyperedge in $M$. 
So, every hyperedge $e \not\subseteq V_M$ is 
certainly properly bicolored. The hyperedges 
which are not properly bicolored in the first bicoloring are 
subsets of $V_M$. 
Then \ref{algo:matching} calls a 
recursive 
function $recMBC$ with the matching $M$ as an argument. 
All the subsequent bicolorings are performed by
$recMBC$.
The function $recMBC$
splits the hyperedges in $M$ into two sets, $M1$ and $M2$
with $|M1|=\lceil \frac{|M|}{2}\rceil$ and $|M2|=\lfloor \frac{|M|}{2}\rfloor$.
\ref{algo:matching} colors the vertices of every hyperedge in $M1$ and $M2$ with colors 0 and
1, respectively.
This gives the second bicoloring.
Note that every hyperedge of $G$ that shares at least one vertex 
each with a hyperedge in $M1$ and a hyperedge in $M2$, 
is properly bicolored.
Now, vertices of hyperedges in $M1$ and $M2$ can be 
colored independently in the subsequent bicolorings.
The function $recMBC$ is invoked recursively on $M1$ and $M2$, 
separately.
Note that $recMBC$ terminates when its argument has a single hyperedge; 
such a hyperedge can be bicolored using a single bicoloring. 

We analyze the number of bicoloring generated by the algorithm \ref{algo:matching}
as follows.
After the first bicoloring, the problem size 
is $|M|$ and the problem size gets halved in each subsequent bicoloring
step. So, after $\log |M|$ bicolorings, the problem reduces to bicoloring of a
single hyperedge, which can be done using a single bicoloring.
We summarize our result in the following theorem.


\begin{theorem}
\label{thm:matching}

For any $k$-uniform hypergraph $G(V,E)$, 
$\chi^c(G) \leq \log |M| +2$, where $M$ is a maximal matching of $G$.
Algorithm \ref{algo:matching} computes such a bicoloring cover for $G(V,E)$ in 
$O(n\log |M|)$ time. 

\end{theorem}


\begin{algorithm}[!htb]
\SetAlgoRefName{$HBC$}
\KwData{$k$-uniform hypergraph $G(V,E)$ with $|V|=n$ and a hitting set $H$}
\KwOut{Set $X$ of bicolorings of size $|X|= \log \lceil \frac{|H|}{k-1}\rceil + 1 $}

Color every vertex in $H$ with color 1 and the
rest of the vertices with color 2\;
Let $G'(H,E')$ be a hypergraph defined 
on the vertices of $H$, and 
$E'$ be all the hyperedges that are monochromatic 
after the first bicoloring\;
$KnCover$($G'$)\;
\caption{Computing bicoloring cover using a hitting set $H$}
\label{algo:hittingset}
\end{algorithm}

%
%

%

Let $H$ be a hitting set of the hypergraph $G(V,E)$.
We propose an algorithm \ref{algo:hittingset} for computing a
bicoloring cover of $G$ using $H$.
\ref{algo:hittingset} takes the hypergraph $G(V,E)$ and the 
hitting set $H$ as inputs and produces a bicoloring cover $C1$.
In the first bicoloring, \ref{algo:hittingset} colors every vertex in $H$ 
with color 0 and all the remaining vertices with the 
color 1. So, the 
hyperedges which are monochromatic in the 
first coloring are subsets of $H$. 
Let $G'(H,E')$ be a hypergraph on the vertices 
of $H$, $E'$ be all the 
hyperedges that are monochromatic after the first bicoloring.
\ref{algo:hittingset} invokes algorithm $KnCover$ on hypergraph 
$G'(H,E')$ to properly bicolor the hyperedges of $G'$.
By Corollary \ref{cor:kncover}, 
we know that $KnCover$ computes a bicoloring cover for $G'$ 
consisting of $\lceil \log  \frac{|H|}{k-1}\rceil$ bicolorings.
%
These $\lceil \log  \frac{|H|}{k-1}\rceil$ bicolorings  
combined with the first bicoloring gives the desired bicoloring cover for $G$.
So, we get the following theorem.

\begin{theorem}\label{thm:hittingset}

For any $k$-uniform hypergraph $G(V,E)$, 
$\chi^c(G) \leq \lceil \log  \frac{|H|}{k-1}\rceil + 1$, where 
$H$ is a hitting set of $G$. 
Algorithm \ref{algo:hittingset} computes such a bicoloring cover for $G(V,E)$ in 
$O(n\log \frac{|H|}{k-1})$ time. 

\end{theorem}

As the union of vertices of some maximal matching $M$ gives a hitting 
set, replacing $|H|$ by $|M|k$, yields the same bound as in Theorem \ref{thm:matching}. As the effectiveness of the algorithm followed in proof of Theorem \ref{thm:matching} depends on the size of 
the maximal matching, finding the smallest maximal matching is useful.


\section{Approximating bicoloring covers}
\label{sec:approx}

Lov\'{a}sz \cite{lovasz1973coverings} showed that the decision problem 
of bicolorability of hypergraphs is NP-complete.
Feige and Killian \cite{Feige1998187} showed that if NP does 
not have efficient randomized algorithms i.e., $NP \not\subset ZPP $, 
then there is no polynomial time algorithm for 
approximating the chromatic number of an $n$-vertex 
graph within a factor of $n^{1-\epsilon}$, 
for any fixed $\epsilon > 0$. Using the above result, 
Krivelevich \cite{Krivelevich20032} demonstrated that for any fixed 
$k\geq 3$, it is impossible to approximate the chromatic number of 
$k$-uniform graphs on $n$ vertices within a factor of 
$n^{1-\epsilon}$ for any fixed $\epsilon >0$, in time polynomial in $n$. 
In Section \ref{subsec:additive}, 
we show that it is impossible to approximate the bicoloring cover of 
$k$-uniform hypergraphs on $n$ vertices within an additive factor 
of $({1-\epsilon}) \log n $ for any fixed $\epsilon >0$, in  
time polynomial in $n$. 
We also 
design approximation algorithms for 
computing bicoloring covers 
in Section \ref{subsec:ratio} 
using the methods developed in \cite{Krivelevich20032}.

\subsection{Inapproximability of the computation 
of \texorpdfstring{$\chi^c(G)$}{xg}}
\label{subsec:additive}

Let $G(V,E)$ be a $k$-uniform hypergraph, and
$\chi^c(G)$ and $\chi(G)$ be the bicoloring cover number and 
the chromatic number of $G$, respectively. 
Assume that $G$ has a bicoloring cover of size $x$ 
i.e., $\chi^c(G) \leq x$. 
By Theorem \ref{thm:chromaticrelation}, $\chi(G) \leq 2^x$. 
Let $R$ be an algorithm that computes a bicoloring 
cover of size $x$ for graph $G$. 
Suppose $R$ is a 
$\alpha$-additive approximation algorithm i.e., for any input instance 
$G$, the size of the computed bicoloring cover 
$x \leq \chi^c(G)+ \alpha$. 
Then, 
using $R$ we can design an approximation algorithm for 
proper coloring of $G$ using 
$2^x \leq 2^{\chi^c(G)} 2^{\alpha} < \chi(G)2^{\alpha+1} $ colors. 
However, in \cite{Krivelevich20032}, it is established that no 
polynomial time algorithm can approximate $\chi(G)$ within a 
factor of $n^{1-\epsilon}$, for any fixed $\epsilon>0$. 
So, setting $2^{\alpha+1}=n^{1-\epsilon}$, 
we get $\alpha =(1-\epsilon) \log n -1$. 
Therefore, we have the following theorem. 

\begin{theorem}

Under the assumption that $NP \not\subset ZPP$, 
 no polynomial time algorithm can approximate the bicoloring cover number $\chi^c(G)$ for $n$-vertex $k$-uniform hypergraph $G(V,E)$ within an additive approximation factor of $(1-\epsilon) \log n -1$, for any fixed $\epsilon > 0$.

\end{theorem}

\subsection{An approximation algorithm for computing bicoloring covers}
\label{subsec:ratio}

Krivelevich and Sudakov \cite{Krivelevich20032}  have developed an algorithm $D(G,p)$ that takes a $n$-vertex $k$-uniform hypergraph $G(V,E)$, and a integer $p \geq \chi(G)$ as inputs, and computes a proper coloring of the hypergraph $G$.
The algorithm $D(G,p)$ uses two algorithms $C_1$ and $C_2$ 
that properly 
color the hypergraph $G$ using at most $8 n^{1-\frac{1}{(k-1)(p-1)+1}}$ 
and $\frac{2n \log p}{\log n}$ colors, respectively if $p \geq \chi(G)$. 
In order words, $D(G,p)$ succeeds in computing an approximate proper 
coloring if $p \geq \chi(G)$.
Since the actual value of chromatic number is not known a priori, $D(G,p)$ 
is executed  with all possible integral values 
of $p$ in the 
range $1$ through $|V|$.
So, the approximation ratio for $\chi(G)$ 
using $D(G,p)$ is $\min\{\frac{8 n^{1-\frac{1}{(k-1)(p-1)+1}}}{p}, 
\frac{\frac{2n \log p}{\log n}}{p}\}$. 
Krivelevich and Sudakov 
use a value of $p=\frac{1}{(k-1)}\frac{\log n}{\log \log n}$ so that both 
the terms in the minimization are of the same order, achieving the 
 approximation ratio of $O(\frac{n (\log \log n)^2}{(\log n)^2}$).

In order to compute a bicoloring cover where the number of bicolorings is within 
a good approximation factor with respect to
$\chi^c(G)$, we use a similar idea and the algorithms of
Krivelevich and Sudakov. 
From Lemma \ref{lemma:ch1}, we know that
$\chi(G) \leq 2^{\chi^c(G)}$. Suppose we invoke $D(G,p)$, where $p=2^s$ and $s \geq \chi^c(G)$.
Then, 
the algorithms $C_1$ and $C_2$ properly color 
the hypergraph $G$ using at most 
$8 n^{1-\frac{1}{(k-1)(2^{s}-1)+1}}$ and $\frac{2ns}{\log n}$ colors, 
respectively.
However, we do not know the value of $\chi^c(G)$ to begin with. 
As we know that $\chi^c(G) \leq \lceil \log 
\frac{n}{k-1} \rceil$ 
(see Theorem \ref{thm:complete}), 
we run $D(G,2^s)$ with all possible values of $s$ in the 
range $1$ through $\lceil \log \frac{n}{k-1} \rceil$ and 
choose the minimum value of $s$ for which $D(G,2^s)$ outputs 
a proper coloring.
From this proper coloring, we can compute a bicoloring 
cover using the reduction stated in the proof of 
Lemma \ref{lemma:ch2}. 
Let $C_{12}$ be the algorithm that (i) takes a $k$-uniform hypergraph $G(V,E)$ 
and an integer $s$ as inputs, (ii) runs $D(G,2^s)$ for 
different values of $s$, and (iii) computes a 
bicoloring cover from the proper coloring output 
of $D(G,2^s)$. 
From Lemma \ref{lemma:ch2}, it is clear that 
$C_{12}$ produces a bicoloring cover of size 
$\min\Big(\log (8 n^{1-\frac{1}{(k-1)(2^s-1)+1}}), \log (\frac{2ns}{\log n})\Big)$.
So, the approximation ratio for algorithm $C_{12}$ is at 
most $\min\Big(\frac{\log (8 n^{1-\frac{1}{(k-1)(2^s-1)+1}})}{s}, 
\frac{\log (\frac{2ns}{\log n})}{s}\Big)$.
We choose the value of $s$ that makes both the terms 
of the same order. Setting $s=\log (\frac{1}{k-1}\frac{\log n}{\log \log n})$, 
the 
the first term becomes 
$\log (8 n^{1-\frac{1}{(k-1)(2^s-1)+1}})=
\log (8 n \allowbreak n^{-\frac{\log \log n}{\log n - (k-2)\log \log n}})\leq 
\log (8 n n^{-\frac{\log \log n}{\log n}})\allowbreak=\log (8 n n^{-\log_{n}\log n})=
O(\log(\frac{n \log \log n}{\log n}))$. 
The 
second term becomes
$\log(\frac{2ns}{\log n})\allowbreak=  
O(\log(\frac{n \log \log n}{\log n}))$. 
Therefore, 
$C_{12}$ has an approximation ratio of 
$O(\frac{\log n+ \log \log \log n-\log \log n}{\log \log n-\log \log \log n} 
\allowbreak) = O(\frac{\log n}{\log \log n-\log \log \log n}-1)$.
We have the following theorem.

\begin{theorem}\label{thm:sudakovapp}
For any  $n$ vertex $k$-uniform hypergraph $G(V,E)$, the bicoloring cover number $\chi^c(G)$ is $\allowbreak O(\frac{\log n}{\log \log n-\log \log \log n})$ approximable.
\end{theorem}

\section{Lower bounds for the bicoloring cover number}
\label{sec:independentset}

In this section we study the relationship between the bicoloring cover 
number, independent sets and related concepts. 
Throughout this section, all the approximation ratios are determined with 
respect to Algorithm \ref{algo:hittingset}.
In Section \ref{subsec:independentset},
we develop certain relationships between 
the bicoloring cover number $\chi^c(G)$, 
the chromatic 
number 
$\chi(G)$,
and introduce a new parameter $\gamma(G)$, 
which we call the {\it cover independence 
number}. 
We demonstrate a better approximation ratio for 
$\chi^c(G)$ 
for hypergraphs where there is large separation between 
$\alpha(G)$ (the {\it independence number}) and $\gamma(G)$. 
In Section \ref{subsec:separation}, we 
demonstrate examples of $k$-uniform hypergraphs $G(V,E)$,
where there is a large separation
between $\gamma(G)$ and $\alpha(G)$.
In Section \ref{subsec:clique}, 
using a probabilistic argument we demonstrate
the existence of hypergraphs with an arbitrarily large  
gap between $\omega(G)$ and $\chi^c(G)$. This  
shows that the 
lower bound for $\chi^c(G)$ as in 
Theorem \ref{thm:complete}
for arbitrary $k$-uniform hypergraphs, 
is not tight.

\subsection{Independence number, cover independence number and
the bicoloring cover number}
\label{subsec:independentset}

A set $I$ of vertices of any hypergraph $G(V,E)$ is 
called an {\it independent} 
set if there is no hyperedge of $G$ in $I$ 
i.e., for no hyperedge $e \in E(G)$, $e \subseteq I$. 
The maximum size of any 
independent set is called the {\it independence number} $\alpha(G)$. Note that $\chi(G) \geq \frac{|V|}{\alpha}$(\cite{MR2002}). Combined with Lemma \ref{lemma:ch1}, we have the following observations.

\begin{observation}\label{obs:independentsetrel}
For a $k$-uniform hypergraph $G$, $\chi^c(G) \geq \log \lceil\frac{|V(G)|}{\alpha (G)}\rceil$.
\end{observation}

\begin{proposition}\label{prop:1}
For a $k$-uniform hypergraph $G$, $\chi^c(G)$ can be approximated in polynomial time by a 
ratio factor $\frac{1}{1-t}$ algorithm if $\alpha(G)=n^t$, where $t<1$.
\end{proposition}

\begin{proof}
Algorithm 
\ref{algo:hittingset}
computes a bicoloring cover of size 
$\lceil \log \frac{|H|}{k-1}\rceil + 1 $ in polynomial time, 
where $H$ is a hitting set for $G$ (see Theorem \ref{thm:hittingset}). 
Following Observation \ref{obs:independentsetrel}, we 
note that the approximation ratio is at 
most $\frac{\log |H|- \log{(k-1)}+1}{\log |V(G)|-\log \alpha(G)}$, 
which is at most $\frac{1}{1-t}$ if $\alpha(G)=n^{t}$ and $t < 1$.
\qed
\end{proof}

From Observation \ref{obs:independentsetrel} 
we note that the bicoloring cover 
number $\chi^c(G)$ is lower bounded by  
$\log \frac{|V|}{\alpha(G)}$.
We introduce the notion of {\it cover independence} in 
Section \ref{subsubsec:notioncovind}, 
and show in 
Theorem \ref{thm:independentsetrel1}
of Section
\ref{subsubsec:lbbasedgamma} 
that $\chi^c(G)$ is lower bounded by 
$\log \frac{|V|}{\gamma(G)}$, where $\gamma(G)$ is the
{\it cover independence} number. 
Further, in Section \ref{subsec:separation}
we construct hypergraphs called {\it cover friendly} hypergraphs, 
where the values of $\alpha(G)$ and
$\gamma(G)$ are 
widely separated.
Observe that Theorem \ref{thm:independentsetrel1} yields a better lower
bound for $\chi^c(G)$ than that given by Observation \ref{obs:independentsetrel}. 


\subsubsection{The notion of cover independence}
\label{subsubsec:notioncovind}

There can be multiple sets of bicolorings of size 
$\chi^c(G)$ that cover $G$. 
Let $w$ be the number of distinct (labeled) bicoloring covers 
of size $\chi^c(G)$, where 
$w \leq 2^{n\chi^c}$. 
Let the set  $C=\{C_1,...,C_w\}$ be
the set of all the 
bicoloring covers of size $\chi^c(G)$ i.e., the set of all the optimal bicoloring covers for $G$.
Let $C_i=\{X_1^i,...,X_{\chi^c(G)}^i\}$, $1\leq i \leq w$, be a 
bicoloring cover of size $\chi^c(G)$, 
where $X_j^i$ denotes the $j^{th}$ bicoloring of vertices of $G$
in the $i$th bicoloring cover $C_i$.
The $\chi^c(G)$ bicolorings in $C_i$ define 
a {\it color bit vector $B_v^i$} of $\chi^c(G)$ bits for
each vertex $v \in V$, 
where the $j^{th}$ bit of $B_v^i$ corresponds to the color of 
$v$ in the $j^{th}$ bicoloring $X_j^i$ in $C_i$.
Consider the partition 
${\cal P}_i=\{V_1^i,V_2^i,...,V_q^i\}$ 
of the vertex set $V$ of $G(V,E)$ 
such that vertices 
$u$ and $v$ belong to the same part, say $V_k^i$, if and only if
$B_u^i$ is identical to $B_v^i$.
The partition
${\cal P}_i$ is called 
a {\it canonical partition} of $V$ due 
to the optimal bicoloring cover $C_i$.
Note that $q$ is 
the number of distinct color bit 
vectors, determining the number of parts in the above partition.
Let $\gamma_i(G)$ be the size of the largest set of vertices that receive 
the same color bit vector for the bicoloring cover $C_i$.
We define

\begin{align*}
\gamma_i(G)=  \max_{\substack{
            1\leq k \leq q}}
     		|V_k^i|,
\end{align*}

We also define 
\begin{align*}
\gamma(G)=  \max_{\substack{
            1\leq i \leq w}}
     		\gamma_i(G),
\end{align*}

We call $\gamma(G)$ the {\it cover independence number} 
of the $k$-uniform hypergraph
$G(V,E)$. 
Any optimal bicoloring cover 
$C_i$ of the hypergraph
$G(V,E)$ with $\gamma_i(G)=\gamma(G)$ is called a {\it witness}
for $G$.

\subsubsection{Interpreting the parameter 
\texorpdfstring{$\gamma(G)$}{r} with examples}

We know that there is a unique bipartition for a connected 
bipartite graph $G(V,E)$, where each edge of $G$ has 
one vertex in each part.  
This unique bipartition yields a bicoloring that covers all 
edges of $G$.
The value of $\gamma(G)$ 
for such graphs is the size of the larger part in this bipartition.
%


\begin{figure}[ht]
\centering
\includegraphics[scale=0.6]{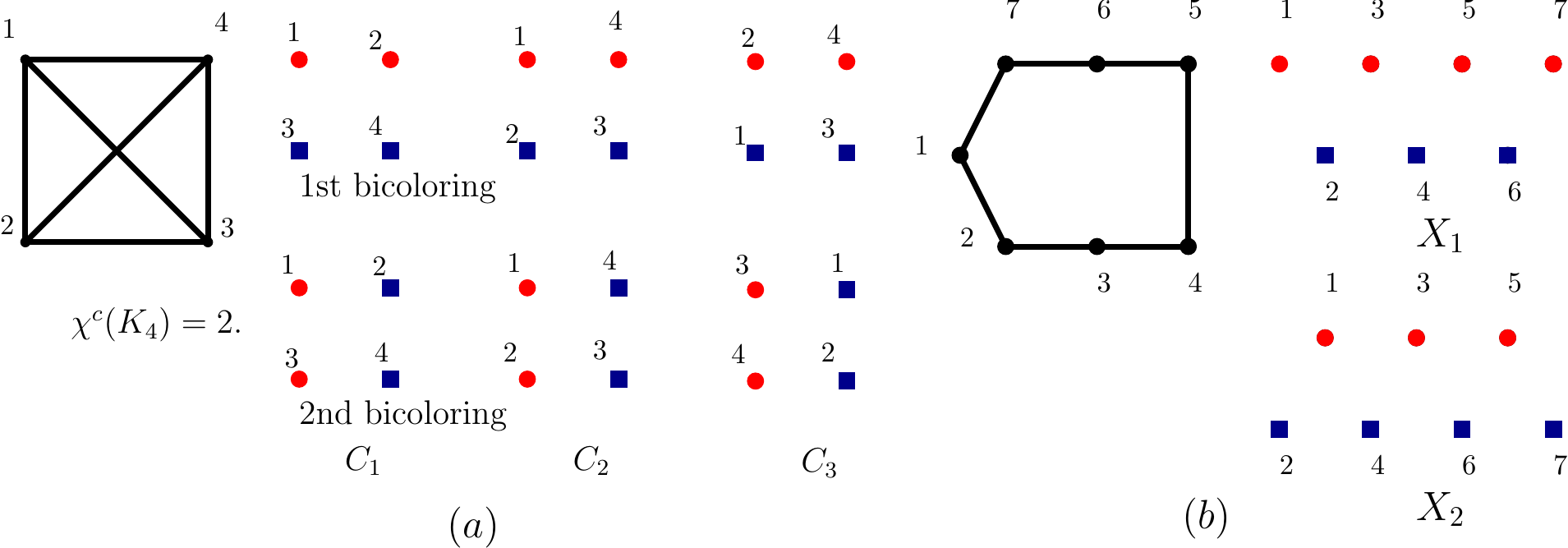}
\caption{$(a)$ $BC_1$, $BC_2$ and 
$BC_3$ denote three distinct bicoloring covers of 
size 2 for $K_4$ where $\gamma(K_4)=1$. 
$(b)$ 
Bicoloring cover $C= \{X_1,X_2\}$ for $C_7$, where
$X_1=red\{1,3,5,7\},blue\{2,4,6\}$,
$X_2=red\{1,3,5,\},blue\{2,4,6,7\}$ and $\gamma(C_7) 
\geq |\{1,3,5\}|=3$.}
\label{fig:gammaex1}
\end{figure}

Note that $\gamma(G)$ for a complete graph is 1. 
Observe that in a complete graph there is an edge
between each pair of vertices. So, a
bicoloring cover $C$ must color vertices in
every pair of vertices with different colors in 
at least one bicoloring $C_i$ of the cover $C$. 
For instance, consider the bicoloring covers of 
$K_4$. From Figure \ref{fig:gammaex1}, it is clear that the 
size of the 
largest set of vertices colored with same color in both the 
bicolorings is 1, in all the three bicoloring 
covers $BC_1$, $BC_2$ and $BC_3$ i.e., 
$\gamma_1(K_4)=\gamma_2(K_4)=\gamma_3(K_4)=1$. 
Therefore, $\gamma(K_4) \geq 1$. 
To see that $\gamma(K_4)<2$, 
observe that if any pair of vertices (say vertices 1 and 2), 
are colored with the same color in both the bicolorings, 
then the edge \{1,2)\} remains uncovered by the set of bicolorings. 

For an odd cycle 
$C_n$, $V=\{v_1,...,v_n\}$, $\gamma(C_n)$ is 
$\frac{n-1}{2}$ and $\chi^c(C_n)=2$. 
Since an odd cycle is not bicolorable, 
$\chi^c(C_n) \geq 2$. 
To show that $\chi^c(C_n)=2$, 
we consider a bicoloring $X_1$, 
where every odd vertex is colored 0, 
every even vertex is colored 1. 
The only edge that is not properly colored is 
${v_1,v_n}$. 
A second bicoloring $X_2$, 
which is exactly the 
same as $X_1$ except that 
the color for $v_n$ is 1 in $X_2$. 
Note that 
$X_2$ properly colors $\{v_1,v_n\}$.
So, $C= \{X_1, X_2\}$ is a bicoloring cover for $C_n$. 
So, $\chi^c(C_n)=2$. 
The vertices $\{v_1,v_3,...,v_{n-2}\}$ are colored 
with 0 in both $X_1$ and $X_2$ in $C$. 
Consequently, $\gamma(C_n) \geq \frac{n-1}{2}$ 
(see Fig. \ref{fig:gammaex1}). 
Also, observe that any subset of 
vertices from the odd cycle $C_n$ with greater than
$\frac{n-1}{2}$ vertices
must contain two consecutive vertices.
Therefore, $\gamma(C_n) \leq \frac{n-1}{2}$.

For any bicolorable hypergraph,  
$\chi^c=1$, set $C$ consists of all the 
proper bicolorings of vertices, $\gamma_i$ is the size of the larger of the two color classes of $i^{th}$ proper bicoloring.
$\gamma \geq \gamma_i \geq \frac{n}{2}$. 
For example, consider the bicoloring of $H(V,E)$, where $V=\{1,2,3,4,5\}$, and $E$ consists of all the 3-uniform hyperedges except $\{1,2,3\}$ and $\{1,2,4\}$ (see Fig. \ref{fig:gammaex2}). 
Certainly $H$ is bicolorable with bicolorings $X_1$ and $X_2$: $X_1=red\{1,2,3\}, blue\{4,5\}$, $X_2= red\{1,\allowbreak 2,4\}, blue\{3,5\}$. $\gamma_1=\gamma_2=3$. Coloring any four vertices with same color in a bicoloring does not cover all the hyperedges. Hence $\gamma=3 \geq \frac{5}{2}$.

\begin{figure}[ht]
\centering
\includegraphics[scale=0.6]{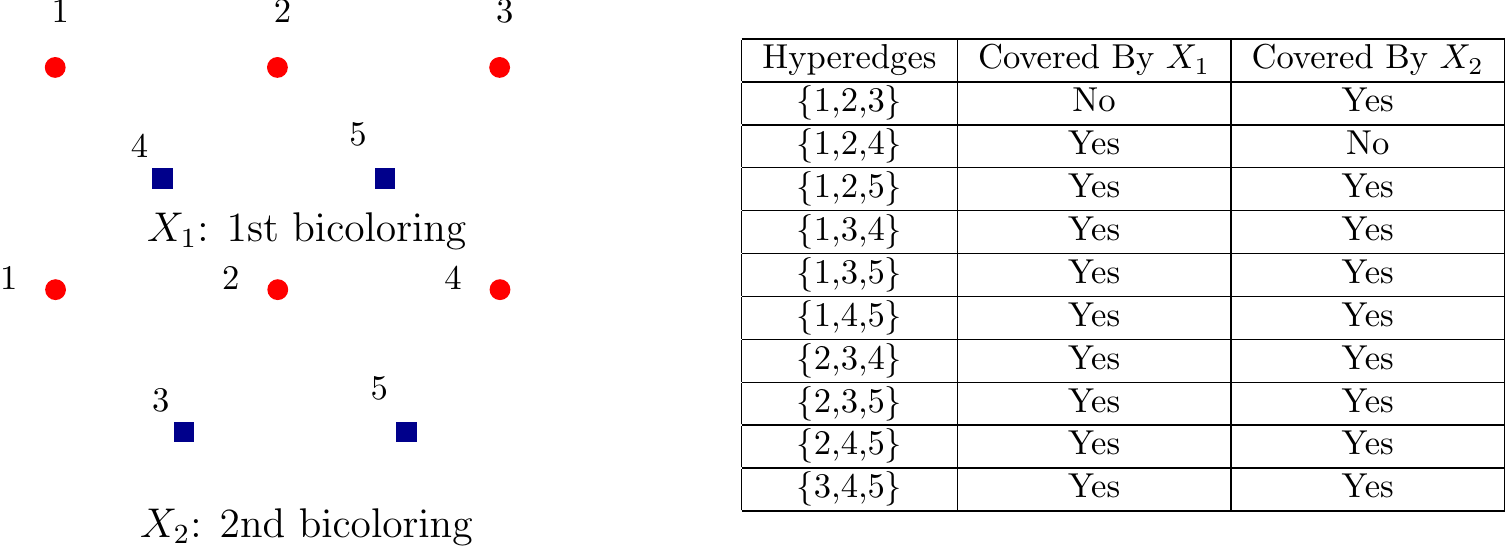}
\caption{ Two bicolorings of $V=\{1,2,3,4,5\}$: $X_1=red\{1,2,3\}$, $blue\{4,5\}$, $X_2= red\{1,2,4\}$, $blue\{3,5\}$. $C= \{X_1,X_2\}$ is a bicoloring cover of $K_5^3$. $\gamma(K_5^3) \geq |\{1,2\}|=2$. 
$C_1=\{X_1\}$ and $C_2=\{X_2\}$ are two distinct bicoloring covers for $H=K_5^3\setminus \{\{1,2,3\}, \{1,2,4\}\}$.
$\gamma_1 \geq |\{1,2,3\}|=3$. $\gamma_2 \geq |\{1,2,4\}|=3$.
$\gamma(H) \geq max(\gamma_1,\gamma_2) =3.$}
\label{fig:gammaex2}
\end{figure}

\subsubsection{A preliminary lower bound for \texorpdfstring{$\gamma(G)$}{r}}

For arbitrary $k$-uniform hypergraph $G$, $\gamma(G) \geq k-1$. 
Assume for the sake of contradiction
that $\gamma(G) = l$ for some $1\leq l \leq k-2$. 
This implies there exists an optimal bicoloring cover 
$C_i$ such that $\gamma_i(G)=l$, and there does not exist any optimal
bicoloring cover $C_j$ 
with $\gamma_j(G)=l+1$. 
We arrive at a contradiction by showing 
that there exists an optimal bicoloring 
cover $C_j$ with $\gamma_j(G)=l+1$. 
Consider $C_i$, a witness for 
$G$, and let $V'$ be the set of vertices such that 
$|V'|=l=\gamma_i(G)=\gamma(G)$, where 
all the vertices in $V'$ receive identical
color bit vectors for the bicoloring cover $C_i$.
Let us move a vertex $s$ from $V \setminus V'$ to 
$V'$ by assigning the color bit vector of $V'$ to $s$, thereby obtaining another set of bicolorings $C_j$ of the same size as $C_i$.
Let $V'_{new}$ = $V' \cup \{s\}$.
If we can show that $C_j$ is a bicoloring cover for the 
hypergraph $G$, then it follows that 
$\gamma_j(G)=|V'_{new}|=l+1$ and we are done. 
Any hyperedge that 
does not contain any vertex from $V'_{new}$
is covered by $C_j$ (using the same bicolorings 
as in $C_i$). 
Note that since $|V'_{new}| \leq k-1$, any hyperedge $e \in E$ that includes vertices from $V'_{new}$,
must contain at least one vertex $t \in V \setminus V'_{new}$.
Since the color bit vectors of $t$ and  $V'_{new}$ are different for $C_j$,
$e$ is properly bicolored by some bicoloring in the set $C_j$.
Since these exhaustive cases include every hyperedge in $G$,  
$C_j$ is a bicoloring cover for $G$, and $\gamma_j(G) =  l+1$, a contradiction. 
Therefore, we have the following theorem.

\begin{theorem}
For any $k$-uniform hypergraph $G$, $\gamma(G) \geq k-1$.
\end{theorem}


\subsubsection{A lower bound for \texorpdfstring{$\chi^c(G)$}{c1} 
and \texorpdfstring{$\chi(G)$}{c}based on \texorpdfstring{$\gamma(G)$}{r}}
\label{subsubsec:lbbasedgamma}

We now study the significance of the cover independence number $\gamma(G)$ 
and its relationship with the bicoloring cover number $\chi^c(G)$ and chromatic
number $\chi(G)$ of a $k$-uniform hypergraph $G(V,E)$. 
Observe that $\gamma(G)$ is the maximum cardinality 
of a subset $S_i\subseteq V$,
of vertices of $G(V,E)$, 
where the color bit vector for vertices in $S_i$ remains invariant 
in the $i$th optimal bicoloring cover.
Consequently, $\gamma(G)$ can be used to lower bound the number of such 
subsets of vertices of $G$, where each such subset of vertices is 
represented by 
its own color bit vector,
as shown in the analysis below. 

Consider any optimal bicoloring cover $\cal{C}$ of a $k$-uniform 
hypergraph $G$. 
The bicoloring cover 
$\cal{C}$ splits $V(G)$ into a canonical partition ${\cal P}$ of at least  
$\lceil \frac{|V(G)|}{\gamma(G)} \rceil$ 
subsets $V_1,..,V_{\lceil \frac{|V(G)|}{\gamma(G)} \rceil},...,V_q$ of $G$, 
where each vertex in any such set $S$ of size at most $\gamma(G)$ 
receives the color bit vector corresponding to the set $S$
for the bicoloring cover $C$.
If for any $i$, $j$, $i < j$, there is no hyperedge that shares 
at least one vertex each with $V_i$ and $V_j$ in $G$, then we
merge $V_j$ into $V_i$. We repeat this process 
for every $i$, $j$, $i < j$, till there is at least one 
hyperedge that shares at least one vertex each 
with $V_i$ and $V_j$. Let this new bicoloring cover 
be ${\cal{C}}1$ and let ${\cal{P}}1= \{V_1,..,V_{p}\}$ be the new
canonical partition of 
the vertices of $G$ due to ${\cal{C}}1$, where $p\leq q$ denotes 
the number of sets in the canonical partition ${\cal P}1$ of $V$ 
due to ${\cal{C}}1$.
For any $i,j,1\leq i <j\leq p$, $V_i$ and $V_j$ are now 
assigned distinct color
bit vectors.
Let ${\cal{C}}1=\{X_1,...,X_{\chi^c(G)}\}$, 
where $X_i$ denote the $i^{th}$ bicoloring
in ${\cal C}1$.
By the definition of $\gamma(G)$, 
$|V_i| \leq \gamma(G)$, $1 \leq i \leq p$, 
so $p \geq \lceil\frac{|V(G)|}{\gamma(G)}\rceil$.


The canonical partition ${\cal{P}}1$ due to ${\cal{C}}1$ can be 
naturally mapped to a complete graph $H$ with (i) the vertex set  
$V(H)=\{1,2,...,p\}$, 
where $p \geq \lceil\frac{|V(G)|}{\gamma(G)}\rceil$, 
and the vertex $i$ corresponds to the 
part $V_i$ of the canonical partition ${\cal{P}}1$, and (ii)
the set $E(H)$ of edges $\{i,j\}$, denoting the existence of 
a hyperedge $e$ of $G$ that shares at least one vertex each with the
corresponding sets $V_i$ and $V_j$ in the canonical partition ${\cal P}1$. 

\begin{proposition}
$H$ is a complete graph.
\end{proposition}

\begin{proof}
According to the definition of 
${\cal{C}}1$, for every $i,j,1 \leq i < j \leq p$, 
there is at least one hyperedge that shares at 
least one vertex each with parts $V_i$ and $V_j$ of the 
canonical partition ${\cal P}1$. 
So, for every $i,j,1 \leq i < j \leq p$, there is an 
edge between vertices $i$ and $j$ in $H$. So, the proposition holds. \qed
\end{proof}

Since $H$ is a complete graph, 
$\omega(H)=p \geq \lceil\frac{|V(G)|}{\gamma(G)}\rceil$.
Using Corollary  \ref{cor:cgraph}, 
we conclude the following lemma.

\begin{lemma}\label{eq:ind}
$\chi^c(H) \geq \log \lceil\frac{|V(G)|}{\gamma(G)}\rceil$.
\end{lemma}

\begin{lemma}\label{claim:ind}
$\chi^c(H) \leq \chi^c(G)$.
\end{lemma}

\begin{proof}
We show that the bicoloring
cover ${\cal{C}}1$ for $G$ can be modified into a
bicoloring cover ${\cal{C}}1'$ for $H$.
We construct ${\cal{C}}1'$ in the following manner.
For each $X_l \in {\cal{C}}1$, we include a bicoloring
$X_l'$ in ${\cal{C}}1'$.
We assign the color of vertices of $V_i$ in $X_l$ to the vertex
$i$ in $X_l'$.
In this construction,
${\cal{C}}1'= \{X_1',...,X_{\chi^c(G)}'\}$, and
$|{\cal{C}}1'|=\chi^c(G)$.
We need to show that ${\cal{C}}1'$ is a valid bicoloring cover for $H$.
Let $e'=\{i,j\} \in E(H)$.
This implies that there exists a
hyperedge $e$ that shares at least one vertex each with
$V_i$ and $V_j$. Suppose $e$ is covered in bicoloring $X_l$ of
${\cal{C}}1$. This implies that $V_i$ and $V_j$ are assigned
different colors in $X_l$.
So, by the construction of $H$, vertices
$i$ and $j$ are colored with different colors in $X_l'$,
thereby covering $e'$.
So, ${\cal{C}}1'$ is a valid bicoloring cover for $H$ and
$\chi^c(H) \leq |{\cal{C}}1'| = \chi^c(G)$.
\qed
\end{proof}

Using Lemma \ref{eq:ind} and Lemma \ref{claim:ind}, we have the following 
theorem.

\begin{theorem}\label{thm:independentsetrel1}

For a $k$-uniform hypergraph $G$, 
$\chi^c(G) \geq \log \lceil\frac{|V(G)|}{\gamma(G)}\rceil$.

\end{theorem}

\begin{corollary}\label{corr:1}
For a $k$-uniform hypergraph $G$, $\chi^c(G)$ can be approximated in polynomial time by a 
ratio factor $\frac{1}{1-t}$ algorithm if $\gamma(G)=n^t$, where $t<1$.
\end{corollary}

\begin{proof}
Algorithm \ref{algo:hittingset} computes a bicoloring cover of size 
$\lceil \log \frac{|H|}{k-1}\rceil + 1 $ in polynomial time for $G$, 
where $H$ is
a hitting set of $G$ (see Theorem \ref{thm:hittingset}). 
Following Theorem \ref{thm:independentsetrel1}, we observe that the achieved
approximation ratio 
is at most $\frac{\log |A|- \log{(k-1)}+1}{\log |V(G)|-\log \gamma(G)}$, 
which is at most $\frac{1}{1-t}$ if $\gamma(G)=n^{t}$ and $t < 1$.
\qed
\end{proof}

Further, we establish the following lower bound for 
$\chi(G)$ based on Lemma \ref{lemma:ch2} and Theorem \ref{thm:independentsetrel1}.

\begin{theorem}\label{thm:independentsetrel2}
For a $k$-uniform hypergraph $G$, $\chi(G) 
\geq  \frac{|V(G)|}{2\gamma(G)}$.
\end{theorem}

\begin{proof}
From Lemma \ref{lemma:ch2}, $\chi^c(G) \leq 
\lceil \log \chi(G) \rceil \leq \log \chi(G) + 1$.
Therefore, $2^{\chi^c(G)}\leq 2\cdot \chi(G)$.
From Theorem \ref{thm:independentsetrel1}, 
$2 \cdot \chi(G) \geq 2^{\chi^c(G)} \geq \frac{|V(G)|}{\gamma(G)}$ 
and the theorem follows.
\qed
\end{proof}

The following proposition establishes the fact that $\alpha(G)$ is at least as large as $\gamma(G)$.

\begin{proposition}\label{prop:alg}
For an arbitrary $k$-uniform hypergraph $G(V,E)$, $\alpha(G) \geq \gamma(G)$.
\end{proposition}

\begin{proof}
We prove the proposition by contradiction. Assume that $\gamma(G) > \alpha(G)$.
Let $\cal{C}$ be one of the bicoloring covers of size $\chi^c(G)$ that produces a subset of vertices $\cal{V}$ of cardinality $\gamma(G)$, such that every vertex in the subset receives the same color in each of the $\chi^c(G)$ bicolorings of $\cal{C}$.
From our assumption, $\gamma(G) > \alpha(G)$, so there must be at least one hyperedge $e \in E$ such that $e \subseteq \cal{V}$.
From the definition of $\cal{V}$, it is clear that $e$ remains monochromatic in all of the $\chi^c(G)$ bicolorings: $\cal{C}$ cannot be a bicoloring cover of size $\chi^c(G)$. This concludes the proof of Proposition \ref{prop:alg}.
\qed

\end{proof}


The gap between $\alpha(G)$ and $\gamma(G)$ becomes a question of great importance for comparing the lower bounds of $\chi^c(G)$ by Proposition \ref{prop:1} and Corollary \ref{corr:1}.
In the following section, we generate an example where $\gamma(G)$ is strictly less than $\alpha(G)$, and also generalize the example to construct a class of hypergraphs where $\alpha(G) \gg \gamma(G)$.


\subsection{Construction of hypergraphs with a large gap between 
\texorpdfstring{$\alpha(G)$}{a} 
and 
\texorpdfstring{$\gamma(G)$}{r}}
\label{subsec:separation}

\subsubsection{A small hypergraph demonstrating the separation between \texorpdfstring{$\alpha(G)$}{a}  and \texorpdfstring{$\gamma(G)$}{r}}
\label{subsubsec:separation}

We need to show that there exists hypergraphs where there is an arbitrary gap between $\alpha(G)$ and $\gamma(G)$. Consider the 3-uniform hypergraph $G1(V,E)$, where  $V = {\cal{V}}_1 \cup {\cal{V}}_2$, ${\cal{V}}_1 = \{1,2,...,8\}$ and ${\cal{V}}_2 = \{9,10,11,12\}$ (see Figure \ref{fig:coverfriendly1}). 
The set of hyperedges $E  = E_1 \cup E_2$ is as follows:
\begin{itemize}
\item $E_1 = \{\{u,v,w\}| 1\leq u < v < w \leq 8\}$, and,
\item $E_2 = \{\{u,v,w\}| 1\leq u < v \leq 8, 9\leq w \leq 12\} \setminus \{\{ 1,5,9\},\{2,6,10\},\{3,7,11\},\{4,8,12\}\}$.
\end{itemize}

\begin{figure}
\centering
\includegraphics[scale=0.5]{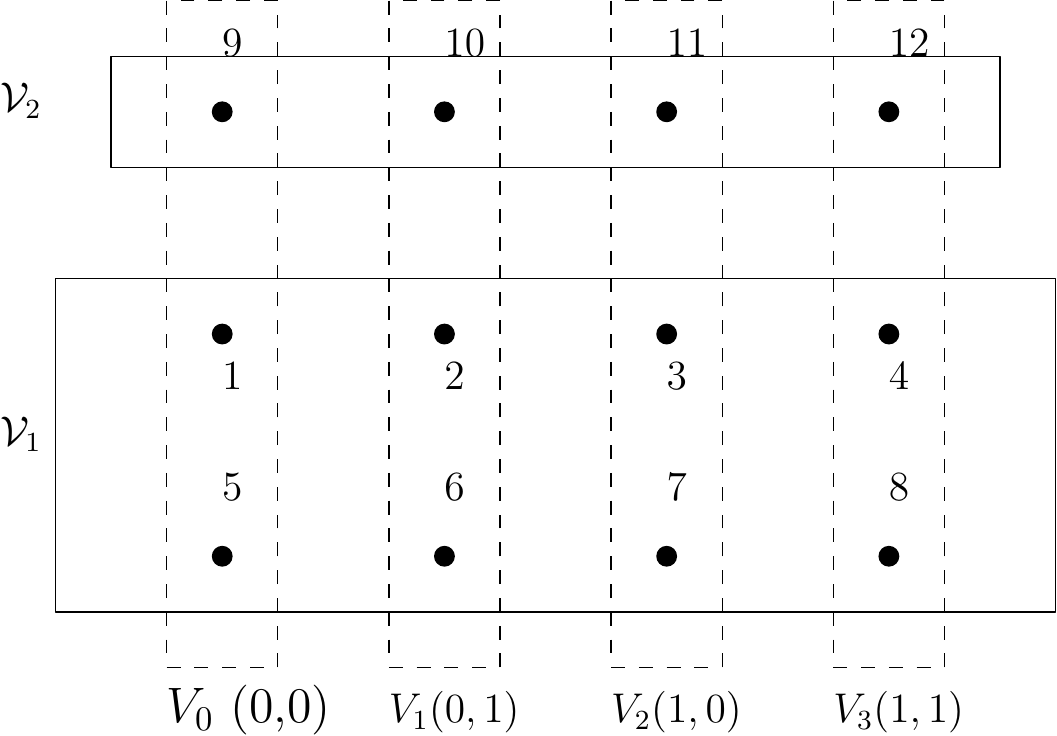}
\caption{Example of a hypergraph $G1(V,E)$ with $\alpha(G1)>\gamma(G1)$. $V = {\cal{V}}_1 \cup {\cal{V}}_2$, where ${\cal{V}}_1 = \{1,2,...,8\}$ and ${\cal{V}}_2 = \{9,10,11,12\}$. $V_0$, $V_1$, $V_2$ and $V_3$
denote the parts with color bit vectors (0,0), (0,1), (1,0) and (1,1) respectively.}
\label{fig:coverfriendly1}
\end{figure}

Observe that ${\cal{V}}_2$ is an independent set as it contains no hyperedges.
Also, observe that $G1$ is not bicolorable since it contains a $K_8^3$ as a subhypergraph (due to hyperedges in $E_1$) and from Theorem \ref{thm:complete}, 
\begin{align}\label{eq:ex1}
\chi^c(G1) \geq \lceil \log \frac{8}{3} \rceil=2.
\end{align}

\begin{lemma}
\label{lemma:alpha}

The independence number $\alpha(G1)$ is five for the hypergraph $G1$. 
Moreover, independent sets of size greater 
that three for $G1$ are obtained by adding at most one vertex 
from ${\cal V}_1$ to subsets of ${\cal V}_2$.

\end{lemma}

\begin{proof}
We prove the above lemma by showing that the maximum 
sized independent sets of $G1$ are $\{\{i,9,10,11,12\}|  
\allowbreak i \in {\cal{V}}_1\}$. Observe that any 
maximum independent set can contain at most two vertices 
from ${\cal{V}}_1$, since any three vertices in  
${\cal{V}}_1$ introduces a hyperedge in $G1$. 
Suppose $u$ and $v$ be any two vertices from 
${\cal{V}}_1$ that are in some maximal 
independent set, $u < v$. If $u \neq v-4$, 
then we cannot add any vertex to that 
independent set. If $u = v-4$, then we can add only 
one vertex $v+4$ to that independent set. Such 
independent sets are of size 3. However, 
restricting only one vertex $u$ from ${\cal{V}}_1$ 
in the independent set, we can add all the vertices of ${\cal{V}}_2$,
generating the independent set $\{u,9,10,11,12\}$.
\qed
\end{proof}

In what follows we show that $\chi^c(G1) \leq 2$, which combined with Inequality \ref{eq:ex1} gives  $\chi^c(G1) = 2$.
Consider the bicolorings of vertices:
\begin{itemize}
\item $X_1=\{0,0,1,1,0,0,1,1,0,0,1,1\}$.
\item $X_2=\{0,1,0,1,0,1,0,1,0,1,0,1\}$, where $j^{th}$ entry in $X_i$ denote the color of the vertex $j$ in $i^{th}$ bicoloring, $1 \leq j \leq 12$.
\end{itemize}

Each vertex $j$, receives a color 
bit vector $(b_{j,2},b_{j,1})$, where 
$b_{j,2}$ and $b_{j,1}$ denote the color of 
$j$ in $X_2$ and $X_1$, respectively.
Let $C1=\{X_2,X_1\}$ be a set of bicolorings 
of $V$. Split $V$ into partition $P=\{V_0,V_1,V_2,V_3\}$ such that
vertex $j$ is added to part $V_0$, $V_1$, $V_2$ or  $V_3$ 
if $j$ receives bits $(0,0)$, $(0,1)$, $(1,0)$, $(1,1)$, 
respectively.
So $V_0=\{1,5,9\}$, $V_1=\{2,6,10\}$, $V_2=\{3,7,11\}$, 
$V_3=\{4,8,12\}$.
Note that each $V_i$, $0\leq i\leq 3$, 
is an independent set, and has a 
distinct color bit vector associated with it due to the bicoloring cover $C1$. 
The construction of $G1$ guarantees that  every hyperedge 
of $G1$ consists of vertices from at least two of the parts.

We wish to show that $\gamma(G1)=3 < \alpha(G)=5$. 
For this purpose, we first show that $C1$ is indeed an optimal 
bicoloring cover and use $C1$ to show that
$\gamma(G1) \geq 3$.
Consider a hyperedge $e=\{u,v,w\} \in E$. 
Without loss of generality, we assume that $u$ and $v$ 
lie in different parts in $P$. Let the color bit vector of $u$ ($v$) be $(a2,a1)$($(b2,b1)$).
By definition of the parts, 
either (i) $a1 \neq b1$, or  (ii) $a1 = b1 \text{ and }a2 \neq b2$. 
If  $a1 \neq b1$, $e$ is covered by $X_1$; 
otherwise, $e$ is covered by $X_2$.
So, $C1$ is a valid bicoloring cover for $G1$ and 
$\chi^c(G1) \leq |C1|=2$. Combined with 
Inequality \ref{eq:ex1}, we conclude that $\chi^c(G1)=2$.
So, $C1$ is an optimal bicoloring cover and by definition, 
$\gamma(G1) \geq \max(\{|{V_i}|| 0 \leq i \leq 3\})=3$.

We now show that $\gamma(G1)$ is 3. 
For the sake of contradiction 
we assume that (i) $C2$ is an optimal bicoloring cover of 
size two for $G1(V,E)$, 
(ii) $\gamma(G1)>3$, and
(iii) $C2$ is a witness for $G1$.
From Lemma \ref{lemma:alpha}, it 
is clear that any independent set of size  
five for $G1$ is obtained only 
by adding any single vertex 
from ${\cal{V}}_1$ to the set ${\cal{V}}_2$. 
From Lemma 
\ref{lemma:alpha}, we also know that 
an independent set of size four in $G1$ is either the set 
${\cal{V}}_2$, or any set with a single vertex from ${\cal V}_1$
and any three vertices from ${\cal V}_2$.
Let any such independent set be called $V_4$.
Since $C2$ is an optimal bicoloring cover, 
the canonical partition $P2$ generated from $C2$ 
consists of at most four mutually disjoint independent sets of $G1$, and since
$C2$ is a witness for $G1$, it must contain $V_4$ as one of the parts.
We define 
$V1'=V\setminus V_4$.
Observe that $V1'$ contains either 7 or 8 vertices 
from ${\cal{V}}_1$ 
i.e., $|V1' \cap {\cal V}_1| \geq 7$. 
If $V_4={\cal V}_2$ or $|V_4|=5$, then
$V1'\subseteq {\cal V}_1$. 
If $V_4$ has one vertex from ${\cal V}_1$, and $|V_4| \geq 4$, 
then
$V1'$ has one vertex from ${\cal V}_2$ 
and seven vertices from ${\cal V}_1$. 
From the construction of $E$, it is clear that vertices from 
$V1'\cap {\cal V}_1$ form 
a $K_7^3$ in $G1$.
Consider the partition of the vertices of $V1'$ into independent sets of $G1$.
Any such independent set of $G1$ can include at most two vertices of 
$V1' \cap {\cal V}_1$; 
three vertices from 
$V1' \cap {\cal V}_1$ always form a hyperedge in $G1$.
Since $V_4$ is already a part in canonical partition $P2$,
there can be at most three more parts in $P2$ as $\chi^c(G1)=2$.
So, $V_4$ has at most 5 vertices and the other at most 
3 parts can include at most $3 \cdot 2=6$ vertices of 
$V1' \cap {\cal V}_1$.
So, at least one vertex $u$ of $V$ ($u \in V1' \cap {\cal V}_1$), 
is not included in the partition $P2$.
Therefore, no such partition $P2$ can include every vertex of $V$. 
So, either $C2$ is not an optimal bicoloring cover, 
or $C2$ is not a witness for $G1$, a contradiction to our assumption.
Consequently, $\gamma(G1)=3$. We conclude that $\alpha(G1)=5 >3 =\gamma(G1)$.

\begin{lemma}
For the hypergraph $G1$, the bicoloring cover number $\chi^c(G1)$ is two, 
and the cover independence number
$\gamma(G1)$ is three. 
\end{lemma}

%
%
%



\subsubsection{An asymptotic construction demonstrating the separation between \texorpdfstring{$\alpha(G)$}{a}  and \texorpdfstring{$\gamma(G)$}{r}}
\label{subsubsec:separation1}

In order to give a general asymptotic
construction, we choose a composite $n$, where $n=p\cdot q$, and
$p$ and $q$ are  integers, $q>p>2$, and 
$q$ is of the form $2^{z}$, $z \in {\mathbb{N}}$.
Let $t = \log_{n}p$.
So, $p=n^t$ and $q=n^{1-t}$. Observe that even keeping $p$
fixed at a certain value, we can indefinitely increase the
values of $n$ and $q$, achieving ever increasing ratio
$\frac {q} {p}=n^{1-2t}$.
Since $p < q$,  it follows that  $0 < t < 0.5$.
Consider the $n$-vertex $k$-uniform hypergraph
$G(V,E)$ 
where $k=n^{t}$.
We design our hypergraph $G$ in such a way that
$\alpha(G)=p+q-2=n^t+n^{1-t}-2$ and
$\gamma(G)=p=n^{t}$.
Let
$V ={\cal V}_1 \cup {\cal V}_2$, where
${\cal V}_1 = \{1,2,...,n-q\}$,
and
${\cal V}_2 = \{n-q+1,n-q+2,...,n\}$.

Let $E_1 = \{\{u_1, ...,u_{k}\}| u_1 < ... <u_k \text{ and } 
\{u_1, ...,u_{k}\} \subset {\cal V}_1\}$, 
$E_2 = \{\{u_1, ...,u_{k}\}| u_k \in {\cal V}_2 \text{, }u_1 < ... < u_{k-1}, \text{ and } 
\{u_1, ...,u_{k-1}\}\subset {\cal V}_1\}$,
$E_3=\{{u_1, ...,u_{k}}| u_1< ...< u_{k}\text{ and for a 
fixed \allowbreak }\allowbreak r, 1 \leq r \leq q, 
\{u_1, ...,u_{k}\} \subseteq V_r  \}$.
Let $E=E_1 \cup E_2 \setminus E_3$. 
Note that ${\cal V}_2$ is an independent set; hyperedges of $G$ are either subsets of ${\cal V}_1$, or
include at most one vertex from ${\cal V}_2$. 

We partition the vertices of $V=\{1,...,n\}$
into $q=n^{1-t}$ parts $\{V_0,...,V_{q-1}\}$,
such that the vertex $i$ is placed in
$V_{(i-1)\mod q}$.
Since $q$ divides $n$,
$|V_r|=p=n^t$,
$0 \leq r \leq q-1$.
So, we get a grid-like arrangement of vertices similar to
that in Figure \ref{fig:coverfriendly1}
with $p = n^t$ rows and $q=n^{1-t}$ columns.
Also, observe that each $V_r$ is an independent set since removal of 
$E_3$ from $E_1 \cup E_2$ removes all hyperedges $e$ that 
lie completely inside a part $V_r$, $0 \leq r \leq q-1$.

\begin{lemma}
\label{lemma:alpha1}

For the hypergraph $G$, $\alpha(G)=p+q-2$. 
Moreover, independent sets of size greater than $p$ for $G$ 
are obtained by adding at most $p-2$ vertices 
from ${\cal V}_1$ to subsets of ${\cal V}_2$.
\end{lemma}
                     

\begin{proof}
Observe that any maximum independent set can 
contain at most $k-1=p-1$ vertices
from 
${\cal V}_1$; otherwise, it introduces at least one hyperedge 
$e \in E_1$. Let $u_1, ...,u_{p-1}$ be any $p-1$ vertices 
from ${\cal V}_1$ that belong to some independent set, say $S$. 
If every $u_j$, $1\leq j\leq p-1$, of this independent set $S$
belongs to the same part 
$V_r$, $1 \leq r \leq q$, then we can add at most one more vertex 
of the same set $V_r$ to the independent set $S$:
this gives an independent set of size $p$. 
Otherwise, if we vertices of $S$ are spread over two parts
$V_r$ and $V_{r'}$, $0\leq r<r'\leq q-1$, 
then adding any other vertex
would give a hyperedge of $G$ and not an independent set. 
However, if we restrict the independent set to include
only $p-2$ vertices $u_1, ...,u_{p-2}$ from ${\cal V}_1$, 
then we can add all the vertices of ${\cal V}_2$ 
to the independent set, thereby generating an independent set 
$\{u_1, ...,u_{p-2},1,...,q\}$, with $p+q-2$ vertices.
\qed

\end{proof}

A lower bound for the bicoloring cover number of $G$ 
is now estimated as follows.
$G$ has a complete $k$-uniform subhypergraph on the vertices of ${\cal V}_1$, 
due to the hyperedges of 
$E$. 
So, using Theorem \ref{thm:complete} 

\begin{align}\label{eq:ind10}
\chi^c(G) \geq \lceil \log (\frac{n-q}{k-1})\rceil =  
\lceil \log (\frac{pq-q}{p-1})\rceil 
=\lceil \log q\rceil.
\end{align}

By construction, $\{V_0,...,V_{q-1}\}$ is a partition of 
$V$ into independent sets. So, using Lemma \ref{lemma:indpset},
$G$ has a bicoloring cover of size $\lceil \log q\rceil$ i.e., 

\begin{align}\label{eq:ind11}
\chi^c(G) \leq \lceil \log q \rceil. 
\end{align}

From Inequalities \ref{eq:ind10} and \ref{eq:ind11}, 
it is clear that the set of bicolorings that partitions 
$V$ into $\{V_0,..., \allowbreak V_{q-1}\}$, is 
a bicoloring cover of optimal size. 
By definition, 
$\gamma(G) \geq \max(\{|{V_r}|| 0 \leq r \leq q-1\})=p$.

We further claim that $\gamma(G)=p=n^t$.
For the sake of contradiction we
assume that (i) $C2$ is an optimal bicoloring cover for $G$ of size 
$\lceil \log q \rceil$, (ii) $\gamma(G)>p$, and (iii) $C2$ is
a witness for $G1$.
From Lemma \ref{lemma:alpha1} we know that
any independent set of 
size strictly greater than $p$ for $G$ must have
at most $p-2$ vertices from ${\cal V}_1$ 
and a subset of vertices from ${\cal V}_2$.
Let any such independent set be $V_q$. Then, $|V_q| \leq p+q-2\leq 
\alpha(G)$ and $|V_q \cap {\cal V}_1 |\leq p-2$.
Since $C2$ is an optimal bicoloring cover, 
the canonical partition $P2$ consists of at 
most 
$2^{\chi^c(G)}$= 
$2^{ \lceil \log q \rceil}=q$ 
independent parts.
Since
$C2$ is a witness for $G1$, 
the canonical partition $P2$ 
of $C2$
must contain $V_q$ as one of its parts.
We define 
$V'=V\setminus V_q$.
Since (i) $|V_q \cap {\cal V}_1| \leq p-2$, and (ii) $|{\cal V}_1| = n-q$, observe 
that $V'$ contains at 
least $n-q-p+2$ vertices 
from ${\cal{V}}_1$
i.e.,  $|V' \cap {\cal V}_1| \geq n-q-p+2$. 
From the definition of $E$, 
the vertices of $V' \cap {\cal V}_1$ form a $K_{n-q-p+2}^p$ 
in $G$.
Consider the partition of the vertices of 
$V'$ into independent sets in $G$.
Any such independent set of $G$ can include at most $p-1$ vertices 
of $V' \cap {\cal V}_1$; 
$p$ vertices from $V' \cap {\cal V}_1$ always form 
a hyperedge in $G$.
Since $V_q$ is already a part in $P2$,
there can be at most $q-1$ more parts in $P2$ as 
$\chi^c(G1)=\lceil \log q \rceil$.
Now, $V_q$ has at most $p+q-2$ vertices and 
the other at most $q-1$ parts can include at 
most $(q-1)(p-1)=pq-q-p+1=n-q-p+1$ 
vertices of $V1' \cap {\cal V}_1$.
So, at least one vertex $u \in V' \cap {\cal V}_1 \subset V$ 
is not included in the partition $P2$.
Therefore, no such partition $P2$ can include every vertex of $V$. 
So, either $C2$ is not an optimal bicoloring 
cover or it is not a witness for $G1$, 
a contradiction to our assumption.
So, $\gamma(G)=p=n^t$. 
Consequently, $\alpha(G)> q = n^{1-t} > p = n^t =\gamma(G)$. This 
concludes the construction of hypergraphs 
where $\gamma(G)=n^t$ and $\alpha(G)> n^{1-t} >\gamma(G)$.

As discussed above, there exists $k$-uniform hypergraphs 
$G(V,E)$ where $\alpha(G) \geq n^{1-t}$, whereas $\gamma(G)=n^t$, 
for some small fraction $0<t <0.5$. 
We call this special class of hypergraphs {\it cover friendly} 
hypergraphs. We have the following lemma.

\begin{lemma}
For the cover friendly $n$-vertex $n^t$-uniform hypergraph $G$, 
$\chi^c(G)=
\lceil \log n^{1-t} \rceil$, and $\gamma(G)
=n^t$, where $t<0.5$.
\end{lemma}
%

Using the results from the asymptotic 
construction of $G$
in Section
\ref{subsubsec:separation1}, 
note that $G1$ 
in Section \ref{subsubsec:separation} (see Figure \ref{fig:coverfriendly1}),
is an instance of $G$ where (i) $|V|=n=12$, 
(ii) $p=3$, (iii) $q=4$, (iv) $t = \log_n{p}=\log_{12}3$, 
(v) $\alpha(G1)=p+q-2=5$, 
(vi) $\chi^c(G1)=\lceil \log q \rceil=2$, and (vii) $\gamma(G1)=p=3$.
We summarize 
the general construction in the following theorem.

\begin{theorem}

Let $G(V,E)$ be a $n$-vertex $k$-uniform hypergraph, 
where $n = pq$, $2 < p < q$, such that $q$ is of the form 
$2^z$, $p,z \in {\mathbb{N}}$,and $k=n^t$. 
Let $t= \log_n p$.
Let $V ={\cal V}_1 \cup {\cal V}_2$, where
${\cal V}_1 = \{1,2,...,n-q\}$,
and
${\cal V}_2 = \{n-q+1,n-q+2,...,n\}$.
Let $E=E_1 \cup E_2 \setminus E_3$, where  
$E_1 = \{\{u_1, ...,u_{k}\}| u_1 < ... <u_k \text{ and } 
\{u_1, ...,u_{k}\} \subset {\cal V}_1\}$, 
$E_2 = \{\{u_1, ...,u_{k}\}| u_k \in {\cal V}_2 \text{, }u_1 < ... < u_{k-1}, \text{, and } 
\{u_1, ...,u_{k-1}\}\subset {\cal V}_1\}$,
$E_3=\{{u_1, ...,u_{k}}| u_1< ...< u_{k}\text{, and for a 
fixed } r,\text{ } 1 \leq r \leq q, \text{ }
\{u_1, ...,u_{k}\} \subseteq V_r  \}$.
Then, \\(i) $ \chi^c(G)= \lceil \log n^{1-t} \rceil$, 
(ii) $\alpha(G) = n^{1-t}+n^t-2$, and (iii) $\gamma(G)=n^t$.

\end{theorem}


For cover friendly hypergraphs, 
using Corollary \ref{corr:1}, we get an 
approximation ratio of $\frac{1}{1-t}$ for approximating $\chi^c(G)$. 
However, using Proposition \ref{prop:1}, 
we get an approximation ratio of at least $\frac{1}{t}$. 
So, we get an improvement (reduction) in 
approximation ratio for $\chi^c(G)$ by a 
factor of at least $\frac{1-t}{t}$, using the properties of $\gamma(G)$. 
Moreover, a constant approximation ratio  of $\frac{1}{1-t}$ for approximating $\chi^c(G)$ is guaranteed for cover friendly hypergraphs as opposed to approximation ratio of $\allowbreak O(\frac{\log n}{\log \log n-\log \log \log n})$ for general hypergraphs given by Theorem \ref{thm:sudakovapp}, exploiting the characteristics of $\gamma(G)$.

\subsection{Clique number and the bicoloring cover number}
\label{subsec:clique}
 
We define {\it clique} number for hypergraphs as follows.
Let $H(V',E')$ be the largest {\it induced subhypergraph} 
of a $k$-uniform hypergraph $G(V,E)$, where
$V' \subseteq V$, $E' \subseteq E$, and $E'\subseteq 2^{V'}$, 
such that every subset 
of $k$ vertices from
$V'$ constitutes a hyperedge in $H$.
We say that $|E'|=\binom{|V'|}{k}$. 
We define the {\it clique} number $\omega(G)$ for hypergraph
$G(V,E)$ as the cardinality of the set $V'$.
Note that $\omega (G)\geq k$ for any $k$-uniform hypergraph.
All non-empty triangle-free undirected graphs have clique number two.
Observe that a non-empty $k$-uniform hypergraph $G(V,E)$ has $\omega(G)=k$ 
provided no induced subhypergraph $G'(S,E')$ of $G(V,E)$, defined on a 
subset $S\subseteq V$ where $|S|=k+1$, 
has all the $k+1$ $k$-sized subsets of $S$ as hyperedges in $E'$. 
Like triangle-free 
graphs, $k$-uniform hypergraphs 
with $\omega(G)=k$ can be quite a rich class of
hypergraphs. We prove the following theorem.

\begin{theorem}\label{thm:clique} 
For any $t\geq 1$, there exists a $k$-uniform hypergraph $G(V,E)$
with $\omega(G)=k$ and $\chi^c(G) > t$.
\end{theorem}

By Theorem \ref{thm:complete}, we know that 
$\chi^c(G) \geq \lceil \log \left(\frac{\omega(G)}{k-1}\right)\rceil$. 
In reality, $\chi^c(G)$ can be arbitrarily far apart 
from $\omega(G)$. Analyzing in a manner similar to that in the 
existential proof of the existence of 
triangle-free graphs with arbitrarily large chromatic 
numbers (see \cite{MR2002}), 
we demonstrate the separation between
$\chi^c(G)$  
and $\omega(G)$ as stated in Theorem \ref{thm:clique}. 
%

A {\it random} $k$-uniform hypergraph $G_{n,p}(V,E)$ is a 
$k$-uniform hypergraph on $n$
labeled vertices $V= [n] = \{1,..., n\}$, in 
which every subset $e \subset V$ of size $|e|=k$ is
chosen to be a hyperedge of $G$ randomly, and 
independently with probability $p$, where
$p$ may depend on $n$. We use $G(V,E)$ or simply $G$ to denote such 
as random hypergraph.
To show the gap between $\omega(G)$ and $\chi^c(G)$, 
we choose a random $k$-uniform hypergraph $G(V,E)$ with 
the value of $p$ set to $n^{-\frac{k}{k+1}}$. 
For showing that 
$\chi^c(G) > t$ 
some arbitrary integer $t$, 
it is sufficient to show that 
$G$ contains no independent set of size 
$\lceil \frac{n}{2^t} \rceil$: we know from
Observation \ref{obs:independentsetrel},  
$\chi^c(G) \geq \log \frac{|V|}{\alpha(G)} > 
\log \frac{n}{\frac{n}{2^t}} = t$.

Let $C_I(G)$ and $C_{\omega}(G)$ denote 
the number of independent sets of 
size $\lceil\frac{kn}{(k+1)2^t} \rceil$ 
and the number of 
complete subgraphs of order $k+1$, respectively,  
in the (random) $k$-uniform
$n$-vertex hypergraph $G(V,E)$.
For any event $x$, let ${\mathcal{E}}(x)$ 
denote the expectation of $x$. We show that
$Prob ({C_I(G)} \geq 1)+Prob ({C_{\omega}(G)} \geq \frac{n}{k+1}) < 1$;
this implies there exists some hypergraph $G(V,E)$ such that ${C_I(G)} = 0$,
as well as ${C_{\omega}(G)} < \frac{n}{k+1}$. 
Then, we delete at most 
$\frac{n}{k+1}$ 
vertices from $G(V,E)$ to generate a new 
hypergraph $G'$ where ${C_I(G')} = 0$, as well as ${C_{\omega}(G')}=0$.

First we show that
$Prob ({C_I(G)} \geq 1)$
with probability strictly less than $\frac{1}{2}$ as follows.
Let $F$ be some set of $\lceil\frac{kn}{(k+1)2^t} \rceil$ vertices
in $G$. 
The probability that $F$ is an independent set is 
$(1-p)^{\binom{\lceil\frac{kn}{(k+1)2^t} \rceil}{k}}$. 
The expectation ${\mathcal{E}}(C_I(G))$ is 
the above probability summed up over all the possible subsets of 
size $\lceil \frac{kn}{(k+1)2^t} \rceil$ in $G$. 
We use the following three known inequalities in our analysis:
(i) $\binom{n}{k} < 2^n$, $0<k<n$,
 (ii) $1-x \leq e^{-x}$, $0\leq x \leq 1$, and
 (iii)  $\binom{n}{k} > (\frac{n}{k})^k$, $0< k <  n$. Therefore,

\begin{align}
\label{ineq:n}
{\cal{E}}(C_I)= &\binom{n}{\lceil \frac{kn}{(k+1)2^t} \rceil}
(1-p)^{\binom{\lceil\frac{kn}{(k+1)2^t} \rceil}{k}}& 
\text{(from the definition of expectation)} \nonumber\\
\leq & \binom{n}{\lceil \frac{kn}{(k+1)2^t} \rceil}
(1-p)^{\binom{\frac{kn}{(k+1)2^t}}{k}} &\nonumber\\
< & 2^n e^{-p \big(\frac{\frac{kn}{(k+1)2^t}}{k}\big)^k} & 
\text{(using (i), (ii) and (iii))}\nonumber\\
= & 2^n e^{-p \big(\frac{n}{(k+1)2^t}\big)^k}&
\end{align}


For a sufficiently large value of $n$ that depends on both 
$k$ and $t$, we can show that 
$2^n e^{-p \big(\frac{n}{(k+1)2^t}\big)^k}$ is strictly less than
$\frac{1}{2}$
(see Appendix \ref{app:clique} for details). 
Now, using Markov's inequality we know that 
$P(C_I(G) \geq 1)\leq {\mathcal{E}}(C_I(G)) < \frac{1}{2}$ for 
sufficiently large values of $n$.

Next, we need to show that the probability 
of existence of complete subhypergraphs of size $k+1$ is small.
Let $W$ be some subset of $k+1$ vertices in $G$. 
Then, $W$ is a complete subgraph with probability $p^{k+1}$. 
The expectation ${\cal{E}}(C_{\omega}(G))$ is given by
${\cal{E}}(C_{\omega})= \binom{n}{k+1} p^{k+1}
< \frac{n^{k+1}}{(k+1)!}\cdot n^{-\frac{k}{k+1} (k+1)}
=\frac{n}{(k+1)!}.
$
Again, using Markov's inequality, $P({C_{\omega}(G)} \geq  \frac{n}{k+1}) 
< \frac{1}{k!} $.
Since $P(C_I(G) \geq 1)+P(C_{\omega}(G) \geq \frac{n}{k+1}) < 1$, 
there exists some graph $G$ such that $C_I(G) = 0$, as well as $C_{\omega}(G) 
< \frac{n}{k+1}$. From each of the at most
$(k+1)$-sized 
complete subhypergraphs, we can remove one vertex each 
(and all the hyperedges incident on it), to 
eliminate that complete subhypergraph of size $k+1$. 
Note that the removal of such vertices and 
corresponding hyperedges cannot increase the size of 
any independent set in $G$.
This transformation results in a subhypergraph 
$G'(V',E')$ of $G$ such that $G'$ does not contain 
any $(k+1)$-sized complete subhypergraphs, and $|V'|\geq n-\frac{n}{k+1}=
\frac{kn}{k+1}$. Moreover, $G'$ does not contain any independent 
set of size $\lceil\frac{kn}{(k+1)2^t} 
\rceil=\lceil\frac{|V'|}{2^t} \rceil$, and therefore $\chi^c(G') > t$.
So, this hypergraph $G'$ has $\omega(G')=k$ but $\chi^c(G')>t$, for any $t>1$,
establishing our claim in Theorem \ref{thm:clique}.

\section{Bicoloring covers for sparse hypergraphs}
\label{sec:probabilistic}

A $k$-uniform hyperedge is rendered
monochromatic with probability 
$\frac{2}{2^{k}}=2^{-(k-1)}$ 
in a random bicoloring of its $k$ vertices. 
If the number of hyperedges $|E|$ in a $k$-uniform hypergraph is
at most $2^{k-2}$, then the probability that some hyperedge is 
rendered monochromatic in a
random bicoloring
is at most $\frac{2^{k-2}}{2^{k-1}} < \frac{1}{2}$. 
Since the probability that none of the hyperedges is monochromatic is at 
least $\frac{1}{2}$, we have the following algorithm for 
computing a bicoloring for $G$. 
Randomly and independently color the vertices of $G$ 
and check whether all the hyperedges are properly bicolored. 
If some hyperedge is rendered monochromatic in the random bicoloring then
repeat the random bicoloring step. 
We can easily verify that 
the expected number of steps of failure is less than two.
Extending similar arguments, we develop the
following relationship between the
number of hyperedges in a $k$-uniform 
hypergraph and the size of its bicoloring cover.

\begin{theorem}
\label{thm:Cover-I}
A $k$-uniform hypergraph $G(V,E)$ with $|E|\leq 2^{(k-1)x-1} $ has a bicoloring cover of size $x$ that can be computed in expected polynomial time.
\end{theorem}

%
%

\begin{proof}
Since all the $x$ bicolorings are random and independent, 
the probability that a specific hyperedge becomes monochromatic 
in each of the $x$ 
bicolorings is $(\frac{1}{2^{k-1}})^x$. 
Choosing the number of hyperedges $|E|\leq 2^{(k-1)x-1}$, 
the probability that some hyperedge becomes monochromatic in each of 
the $x$ bicolorings is strictly less than $\frac{1}{2}$. 
In other words, the 
probability that each of the $|E|$ hyperedges is 
non-monochromatic in one or more bicolorings is at least $\frac{1}{2}$. 
Consequently, the hypergraph has a cover of size $x$ and 
that can be computed by random coloring of vertices in 
expected two iterations. \qed
\end{proof}

Since Theorem \ref{thm:Cover-I} gives only a 
sufficiency condition for a $k$-uniform 
hypergraph to have a bicoloring cover of size 
$x$, it is interesting to estimate the smallest integer $m$  
such that there is no 
bicoloring cover with $x$ bicolorings 
for some $k$-uniform hypergraph with $m$ hyperedges. 
This number $m$ 
is a measure of the
tightness of the sufficiency condition given by 
Theorem \ref{thm:Cover-I}.
We define $m(k,x)$ as the smallest integer
such that there exists a $k$-uniform hypergraph $G$ with $m(k,x)$
hyperedges, 
which does not have a bicoloring cover of size $x$.
%
%
If the number of hyperedges in the $k$-uniform hypergraph is less 
than $m(k,x)$, then it certainly has a bicoloring cover of size $x$. 
In other words, for any hypergraph of size less than $m(k,x)$, 
there exist at least one set of $x$ bicolorings of 
vertices that properly bicolors every hyperedge 
of the hypergraph. However, if the number of hyperedges 
is greater than or equal $m(k,x)$, 
then we cannot guarantee the existence of a 
bicoloring cover of size $x$ for the hypergraph. 
Alternatively, there exist at least one 
$k$-uniform hypergraph of size $m(k,x)$ such that no 
set of $x$ bicolorings can properly bicolor every 
hyperedge in the hypergraph. 
From Theorem \ref{thm:Cover-I}, 
it is obvious that $m(k,x) > 2^{(k-1)x-1}$. 

We note that $8< m(2,3,2) \leq 84$. 
The lower bound is 
given by the proper substitution in 
Theorem \ref{thm:Cover-I}. The 
upper bound is obtained from a 
$K_9^3$ which does not have a 
bicoloring cover of size 2 (see Theorem \ref{thm:complete}).
Computing the exact values of $m(k,x)$ for different 
values of $x$ by brute-force is difficult. 
In order to prove that $m(k,x)=a$, 
for some fixed $x$, $k$ and $a$, 
one may find out at least one 
$k$-uniform hypergraph with $a$ hyperedges that 
does not have a bicoloring cover of size $x$. 
So, one may check every hypergraph with $a$ hyperedges 
for a bicoloring cover of size $x$, which is computationally expensive. 
In order to estimate an upper bound for $m(k,x)$, 
we 
consider a (i) $k$-uniform hypergraph with $k+k^2$ vertices,  
(ii) fix a set $C1$ of $x$ independent bicolorings of the vertex set $V$, and
(iii)
pick  $m$ $k$-uniform hyperedges uniformly, independently and randomly. 
Let the probability that a randomly picked hyperedge $e$ becomes monochromatic in 
a random bicoloring be at least $p$; below we estimate a lower bound for 
$p$ considering random bicolorings. Then, 
the probability that $e$ becomes monochromatic in each of 
the $x$ bicolorings is at least  $p^x$. 
So, the probability that $e$ becomes non-monochromatic in at least 
one of the $x$ bicolorings is at most $1-p^x$.
Since, we are choosing $m$ hyperedges independently, randomly and uniformly,
the probability that every hyperedge from the $m$ 
chosen hyperedges becomes non-monochromatic in at least one of the 
$x$ of bicolorings in $C1$ is at most $(1-p^x)^m$.
Since there are $2^{nx}$ ways to perform the $x$ independent bicolorings,
the probability each of the $m$ chosen hyperedges 
becomes non-monochromatic in at least 
one set of $x$ colorings 
is at most $ 2^{nx} (1-p^x)^m$.
So, if $f(n,x,p,m)=2^{nx} (1-p^x)^m < 1$, then there exists at 
least one set of $m$ hyperedges that cannot be 
covered by any set of $x$ bicolorings. 

Now we estimate a lower bound for $p$, the probability with which 
any randomly picked hyperedge becomes monochromatic in any bicoloring of 
the $k+k^2$ vertices in $V$.
Any bicoloring colors some vertices with color 1 and rest with color 2. Let 
the set of color 1 vertices be of size $a$. Then, the total number of 
monochromatic hyperedges is $\binom{a}{k}+\binom{n-a}{k}$. This 
sum is minimized at $a=\lceil\frac{n}{2}\rceil$.
%
Therefore, the probability that a particular random hyperedge $e$ 
is monochromatic in 
one bicoloring is at least 
$2\frac{\binom{\frac{n}{2}}{k}}{\binom{n}{k}} = 
2*\frac{\frac{n}{2}(\frac{n}{2}-1)...(\frac{n}{2}-k+1)}{n(n-1)...(n-k+1)}> \frac{1}{2^{k-1}}(\frac{n-2k}{n-k})^k =\frac{1}{2^{k-1}}(1-\frac{1}{k})^k$ (since $\frac{\frac{n}{2}}{n}>\frac{\frac{n}{2}-1}{n-1}>...>\frac{\frac{n}{2}-k+1}{n-k+1}>\frac{\frac{n}{2}-k}{n-k}$). 
Let $p= \frac{1}{2^{k-1}}(1-\frac{1}{k})^k$. 
For $k \geq 2$, $(1-\frac{1}{k})^k \geq \frac{1}{4}$. 
We find that 
the expression $f(n,x,p,m)$ is upper bounded by
$2^{nx}(1-(\frac{1}{2^{k+1}})^x)^m
 < 2^{nx}e^{-\frac{m}{2^{(k+1)x}}}$.
The last expression becomes unity when $m$ is set to $2^{(k+1)x} \cdot n \cdot x \ln 2$.
This implies that there exists a hypergraph with $m$ hyperedges such that at 
least one of the $m$ hyperedges remains monochromatic in each set of $x$ bicolorings.
Since $n=k^2+k$, we have 
$2^{(k+1)x} \cdot n \cdot x \ln 2< 2^{(k+1)x} 
\cdot 2k^2 \cdot x \cdot 2=x k^2 2^{(k+1)x+2}$. We state our result in 
the following theorem.

\begin{theorem}
$2^{(k-1)x-1} < m(k,x) \leq x \cdot k^2 \cdot 2^{(k+1)x+2}$.
\end{theorem}


\section{Computing bicoloring covers for hypergraphs with bounded dependency}
\label{sec:deprel}



The {\it dependency} of a hyperedge $e$ in a $k$-uniform 
hypergraph $G(V,E)$,
denoted by $d(G,e)$ 
is the number of hyperedges in the set $E$ 
with which $e$ shares at least one vertex. 
The {\it dependency of a hypergraph} $d(G)$ or simply $d$, denotes the 
maximum dependency of any hyperedge in the hypergraph $G$.
Lov\'{a}sz local lemma  
\cite{loverd1975,Motwani:1995:RA:211390,Moser:2010:CPG:1667053.1667060} 
ensures the existence of a proper bicoloring 
for any $k$-uniform hypergraph provided the 
dependency of the hypergraph 
is upper bounded by $\frac{2^{k-1}}{e}-1$. 
Furthermore, the constructive
version of 
Lov\'{a}sz local lemma by Moser and Tardos 
\cite{Moser:2010:CPG:1667053.1667060} 
enables the computation of a bicoloring of a $k$-uniform 
hypergraph with dependency at most   
$\frac{2^{k-1}}{e}-1$ by a randomized algorithm. 
Chandrasekaran et.al. \cite{chand2013} proposed a derandomization  for local lemma  that computes a bicoloring
in polynomial time.
In what follows, we  
use similar techniques for establishing permissible bounds on the 
dependency of a hypergraph as a function of the size of its 
desired bicoloring cover, and for computing such bicoloring covers. 
The Kolmogorov complexity approach for Lov\'{a}sz local lemma 
leads to a method that can bicolor a hypergraph whose dependency is 
at most $2^k/8$ (see \cite{LF2009}). 


%

Let $P$ be a finite set of mutually independent random variables
in a probability space. 
We consider a finite set $\cal A$ of events, where each event 
$A\in \cal A$ is determined by a subset $S(A) \subseteq P$ of the variables
in $P$.
We say that an event $A\in \cal A$ is violated if 
an evaluation of variables in $S(A)$ results in the 
occurrence of $A$. We have the following lemma due to 
Moser and Tardos \cite{Moser:2010:CPG:1667053.1667060}.

\begin{lemma}
\label{cover_lemma_2}
\cite{Moser:2010:CPG:1667053.1667060}
Let $P$ be a finite set of mutually independent random variables
in a probability space. 
Let ${\cal A}$ be a finite set of events determined by these variables.
For any $A \in {\cal A}$, let $\Gamma(A)$ 
denote the set of all the events in ${\cal A}$ that 
depend on $A$.
If there exists an assignment of reals 
$x : {\cal A} \rightarrow (0, 1)$ such that
$\forall A \in {\cal A} : Pr[A] \leq x(A) \prod_{B \in \Gamma(A)}(1-x(B))$,
then there exists an assignment of values to the variables in
$P$ not violating any of
the events in ${\cal A}$. 
Moreover the Moser-Tardos Sequential Solver algorithm 
\cite{Moser:2010:CPG:1667053.1667060} resamples
an event $A \in {\cal A}$ at most an expected 
$x(A)/(1-x(A))$ times before it finds such
an evaluation. Thus, the expected total number 
of resampling steps is at most
$\sum_{A \in {\cal A}}\frac{x(A)}{1-x(A)}$.

\end{lemma}

In particular, if $\forall A \in {\cal A}$, $x(A)$ is set 
to $\frac{1}{d+1}$ and $P[A] \leq p$, then the premise of 
Lemma \ref{cover_lemma_2} reduces to $ep(d+1) \leq 1$, 
where $d$ is the maximum {\it dependency} 
$\max_{A \in {\cal A}}|\Gamma(A)|$ 
of any 
event $A$ in 
$\cal A$.
So, from 
Lemma \ref{cover_lemma_2} with suitable substitutions, 
we get the following corollary.

\begin{corollary}\label{cor:lll}

Let $P$ be a finite set of mutually independent random variables
in a probability space. Let ${\cal A}$ be a 
finite set of events determined by these variables,  where $m=|{\cal A}|$.
For any $A \in {\cal A}$, let $\Gamma(A)$ denote the set of all the events in ${\cal A}$ that depend on $A$.
Let $d = \max_{A \in {\cal A}}|\Gamma(A)|$.
If 
$\forall A \in {\cal A} :P[A] \leq  p  \text{ and } ep(d+1) \leq 1$, 
then
an assignment of the variables not violating any of
the events in ${\cal A}$ can be computed using expected 
$\frac{1}{d}$ resamplings per event and
expected $\frac{m}{d}$ resamplings in total.

\end{corollary}


\begin{algorithm}[H]
\SetAlgoRefName{$MTC$}
 \KwData{$k$-uniform hypergraph $G(V,E)$ with $d\leq \frac{2^{|X|(k-1)}}{e}-1$}
 \KwResult{Set $X$ of bicolorings of size $|X|$ }
\For{$v \in V$}{
	\For{$i \in \{1,...,|X|\}$}{
			${r_v}^i \leftarrow$ a random evaluation of $v$ in $i^{th}$ bicoloring of $X$\;
		}
}
\While{$\exists A_i \in \mathcal{A}$: $A_i$ happens i.e., every bicoloring in $X$ renders $E_i$ monochromatic}{
		Pick an arbitrary violated event $A_i \in \mathcal{A}$\;
		\For{$ v \in E_i$}{
			\For{$i \in \{1,...,|X|\}$}{
					${r_v}^i \leftarrow$ a random evaluation of $v$ in $i^{th}$ bicoloring of $X$\;
				}
		}
}
\caption{Randomized algorithm for computing a bicoloring cover}
\label{algo:admostar}
\end{algorithm}

In what follows, we use Corollary \ref{cor:lll} and 
an adaptation of the Moser-Tardos algorithm which we 
call $MTC$, 
to compute a bicoloring cover $X$ of $x$ bicolorings, 
for a $k$-uniform hypergraph $G(V,E)$. 
Let the event $A_i$ correspond to the hyperedge $E_i \in E$ becoming 
monochromatic in each of the $x$ random and 
independent bicolorings. The probability $p(A_i)$ is 
at most $\big(\frac{1}{2^{k-1}}\big)^{x}$.
So, using Corollary \ref{cor:lll}, the maximum allowable 
dependency $d$ of the hypergraph $G$ is
$\frac{2^{x(k-1)}}{e}-1$, so that $G$ has a bicoloring cover 
with $x$ bicolorings. 
In order to compute a bicoloring cover for $G$ with $x$
bicolorings, 
where $d(G) \leq \frac{2^{x(k-1)}}{e}-1$, the algorithm $MTC$, repeatedly recolors 
vertices of monochromatic hyperedges, one at a time.
It picks up a monochromatic hyperedge and generates 
$x$ random bits 0/1
for each vertex of the monochromatic hyperedge, one bit for each of the $x$ 
bicolorings. If there are several monochromatic hyperedges 
then $MTC$ picks up any such hyperedge for recoloring 
all its vertices with 
colors 0/1,
for each of the $x$ bicolorings. 
Each such step is called a resampling step,
where one hyperedge gets all its $k$ vertices recolored for each of the $x$ 
bicolorings. 
The correctness of $MTC$ follows from the correctness of the 
Moser-Tardos constructive version of the local lemma; 
the algorithm terminates after generating a bicoloring 
cover with $x$ bicolorings.

Since $MTC$ is a randomized algorithm, it consumes random bits in 
each resampling step. Let $T$ be the total number of resampling steps 
performed. The algorithm $MTC$ 
uses 
$nx+Tkx$ random bits
for computing a bicoloring cover of size $x$,
given $d \leq \frac{2^{x(k-1)}}{e}-1$. 
Here,
$nx$ random bits are
for the initial assignment (one bit per vertex per bicoloring),
$kx$ bits each for each of the $T$ resampling 
steps 
(one bit per vertex of the resampled hyperedge per bicoloring).
We know from Corollary \ref{cor:lll} that the expected number 
of resampling steps is $T=\frac{m}{d}$.
So, the expected number of random bits 
used by the algorithm is  $nx+kx\frac{m}{d}$.
Since $d \leq  \frac{2^{x(k-1)}}{e}-1$, we have 
$x \geq \frac{1}{k-1} \log (e(d+1))$.
Therefore, the expected number of random bits used 
by the algorithm is at least 
$\frac{1}{k-1} \log (e(d+1))(n+k\frac{m}{d})$. 
We summarize these results as the following theorem.

\begin{theorem}\label{thm:genlocallemma}
Let $G(V,E)$ be a $k$-uniform hypergraph, $n=|V|$, $m=|E|$.
Let the dependency of the hypergraph $d(G)$ 
be upper bounded by $\frac{2^{x(k-1)}}{e}-1$, for some $x \in \mathcal{N}$.
Then, there exists a bicoloring cover of size $x$, 
which can be computed by a randomized algorithm 
using $\frac{1}{d}$ expected resamplings per hyperedge and 
$\frac{m}{d}$ resamplings in total, 
using expected $nx+kx\frac{m}{d}$ random bits.
\end{theorem}

We note that Algorithm $MTC$ can be derandomized in the same manner 
as done for the case of bicolorings in \cite{chand2013}. 
As the dependency of the hypergraph grows, 
the bicoloring cover size guaranteed by the local lemma also 
increases. However, from Theorem \ref{thm:complete}, 
we know that for any $k$-uniform graph $G(V,E)$, $|V|=n$, 
$\chi_c(G)\leq \left\lceil \log({\frac{n}{k-1}})\right\rceil$. 
So, the application of this algorithm is practical for 
the case where it guarantees a cover of size 
of at most $\lceil \log({\frac{n}{k-1}})\rceil$. 
We can find the maximum dependency for which this algorithm is 
useful by simply replacing $x$ in the dependency bound as 
$d\leq \frac{2^{x(k-1)}}{e}-1 \leq \frac{1}{e} {2^{ \log({\frac{n}{k-1}})}}^{(k-1)}-1$,
that is
$d \leq  \frac{1}{e}{(\frac{n}{k-1})}^{(k-1)}-1.$

\section{Concluding remarks}
\label{sec:conclusion}

Bounds for bicoloring cover numbers established in
this paper are supported by algorithms that generate the 
bicoloring covers of the corresponding sizes. 
The algorithms and bounds can be generalized for 
multicolorings, where more than two colors are used.  
In such natural extensions to multicolorings, the 
constraint imposed on every hyperedge 
can be relaxed so that
at least $p \geq 2$ vertices of the hyperedge are 
distinctly colored in at least one of the multicolorings.

Throughout the paper, we have used independent bicolorings 
in our probabilistic analysis. 
Whether the use of mutually dependent bicolorings 
would lead to discovery of better bounds for  bicoloring cover numbers, 
remains an open question.
Computing the exact value or approximating the cover independence number $\gamma(G)$ 
remains an open problem.

{\small
\bibliography{REFERENCES}}
\bibliographystyle{plain}
\small

\appendixtitleon
\appendixtitletocon
\begin{appendix}

\section{Estimation of \texorpdfstring{$n$}{n} for Inequality \ref{ineq:n} in Section \ref{subsec:clique}}
\label{app:clique}

From Section \ref{subsec:clique}, we have, $p=n^{-\frac{k}{k+1}}$.
Using Inequality \ref{ineq:n}, we need to choose a value of $n$ that satisfies the inequality $2^n e^{-p \big(\frac{n}{(k+1)2^t}\big)^k} <  \frac{1}{2}$. We proceed as follows.
\begin{align*}
& 2^n e^{-p \big(\frac{n}{(k+1)2^t}\big)^k} <  \frac{1}{2} \\
\Leftrightarrow & 2^{n+1} < e^{p \big(\frac{n}{(k+1)2^t}\big)^k} \\
\Leftrightarrow &(n+1) \log_e 2 < n^{-\frac{k}{k+1}}\big(\frac{n}{(k+1)2^t}\big)^k \\
\Leftrightarrow &(k+1)^k 2^{tk}  \log_e 2 < \frac{n^{\frac{k^2}{k+1}}}{n+1} 
\end{align*}
This inequality can always be satisfied for a sufficiently large value of $n$
$n > ((k + 1)^k *2^{tk+1} log_e 2)^{(k+1)/(k^2)}$.

\end{appendix}

\end{document}